\documentclass[11pt]{article} 

\usepackage{fullpage}
\usepackage{tikz}
\usetikzlibrary{arrows,shapes,snakes,automata,backgrounds,petri}
\usepackage{graphicx,color,float}
\usepackage{epsfig} 
\usepackage{amsmath,amsthm}
\usepackage{amssymb}
\usepackage{mathabx}
\usepackage[framemethod=TikZ]{mdframed}

\usepackage{colortbl}
\usepackage{algorithm2e}
\usepackage{mathtools}
\usepackage{amssymb}
\usepackage{todonotes}
\usepackage{bbm}
\usepackage{comment}

\usepackage{amsfonts}
\usepackage{mathptmx}

\newtheorem{theorem}{Theorem}[section]
\newtheorem{observation}[theorem]{Observation}
\newtheorem{lemma}[theorem]{Lemma}
\newtheorem{claim}[theorem]{Claim}
\newtheorem{remark}[theorem]{Remark}
\newtheorem{definition}[theorem]{Definition}
\newtheorem{corollary}[theorem]{Corollary}

\mdfdefinestyle{MyFrame}{%
    outerlinewidth=0.5pt,
    roundcorner=10pt,
	innertopmargin=5pt,
    innerbottommargin=2pt,
    innerrightmargin=20pt,
    innerleftmargin=20pt,
    }
	
\DeclareRobustCommand{\mmc}{%
\begin{tikzpicture}%
\draw [opacity=0.3](0,0) -- (1ex,0);%
\draw [opacity=0.3](1ex,0) -- (1ex,0.8ex);%
\draw [opacity=0.3](-0.8ex,0.8ex) -- (1ex,0.8ex);

\fill (0,0) circle[radius=1pt];
\fill (-0.8ex,0.8ex) circle[radius=1pt];
\fill [blue] (1ex,0) circle[radius=1pt];
\fill [blue] (1ex,0.8ex) circle[radius=1pt];
\end{tikzpicture}%
}

\DeclareRobustCommand{\mfc}{%
\begin{tikzpicture}%
\draw [opacity=0.3](0,0) -- (1ex,0);%
\draw [opacity=0.3](1ex,0) -- (1ex,0.8ex);%
\draw [opacity=0.3](1.8ex,0.8ex) -- (1ex,0.8ex);

\fill (0,0) circle[radius=1pt];
\fill (1.8ex,0.8ex) circle[radius=1pt];
\fill [blue] (1ex,0) circle[radius=1pt];
\fill [blue] (1ex,0.8ex) circle[radius=1pt];
\end{tikzpicture}%
}

\DeclareRobustCommand{\interone}{%
\begin{tikzpicture}%
\draw (0,0) -- (1ex,0);%
\draw (1ex,0) -- (1ex,1ex);%
\draw (1ex,1ex) -- (0,1ex);
\draw (0ex,1ex) -- (0,0ex);

\draw (-0.5ex,0.5ex) -- (0.5ex,0.5ex);%
\draw (0.5ex,0.5ex) -- (0.5ex,-0.5ex);%
\draw (0.5ex,-0.5ex) -- (-0.5ex,-0.5ex);
\draw (-0.5ex,-0.5ex) -- (-0.5ex,0.5ex);

\end{tikzpicture}%
}

\DeclareRobustCommand{\intertwo}{%
\begin{tikzpicture}%
\draw (0,0) -- (1ex,0);%
\draw (1ex,0) -- (1ex,1.5ex);%
\draw (1ex,1.5ex) -- (0,1.5ex);
\draw (0ex,1.5ex) -- (0,0ex);

\draw (-0.33ex,1.2ex) -- (1.33ex,1.2ex);%
\draw (1.33ex,1.2ex) -- (1.33ex,0.3ex);%
\draw (1.33ex,0.3ex) -- (-0.25ex,0.3ex);
\draw (-0.33ex,0.3ex) -- (-0.33ex,1.2ex);

\end{tikzpicture}%
}

\DeclareRobustCommand{\interthree}{%
\begin{tikzpicture}%
\draw (0,0) -- (1ex,0);%
\draw (1ex,0) -- (1ex,1.5ex);%
\draw (1ex,1.5ex) -- (0,1.5ex);
\draw (0ex,1.5ex) -- (0,0ex);

\draw (0,0) -- (-0.5ex,0ex);%
\draw (-0.5ex,0ex) -- (-0.5ex,1ex);%
\draw (-0.5ex,1ex) -- (1ex,1ex);%
\end{tikzpicture}%
}

\DeclareRobustCommand{\interfour}{%
\begin{tikzpicture}%
\draw (0,0) -- (1.5ex,0);%
\draw (1.5ex,0) -- (1.5ex,0.7ex);%
\draw (1.5ex,0.7ex) -- (0,0.7ex);
\draw (0ex,0.7ex) -- (0,0ex);

\draw (0.7ex,0.7ex) -- (0.7ex,-0.5ex);
\draw (0.7ex,-0.5ex) -- (0,-0.5ex);
\draw  (0,-0.5ex) -- (0,0);

\end{tikzpicture}%
}

\newif\iflong
\longtrue

\newcommand{\Key}[1]{\ensuremath{#1.x}}
\newcommand{\Time}[1]{\ensuremath{#1.y}}
\newcommand{\RR}{\textsc{Pair}}
\newcommand{\OP}{\textsc{Op}}
\newcommand{\UP}{\textsc{Up}}

\newcommand{\GR}{\textsc{Good}(\MFC)}
\newcommand{\BR}{\textsc{Bad}(\MFC)}
\newcommand{\OPT}{\textsc{Opt}}
\newcommand{\TD}{\textsc{TreeDecomposition}}
\newcommand{\PP}{\textsc{Parent}}
\newcommand{\SB}{\textsc{Sibling}}
\newcommand{\Top}{\textsc{Top}}
\newcommand{\RT}{\textsc{TopBlock}}
\newcommand{\CP}{\textsc{Pair}}
\newcommand{\MMC}{\mmc}
\newcommand{\LMMC}{\textsc{L}(\MMC)}
\newcommand{\RMMC}{\textsc{R}(\MMC)}
\newcommand{\MFC}{\mfc}
\newcommand{\LMFC}{\textsc{L}(\BR)}
\newcommand{\RMFC}{\textsc{R}(\BR)}
\newcommand{\R}{\textsc{R}}
\newcommand{\LL}{\textsc{L}}
\newcommand{\UPB}{\textsc{UpperBox}}
\newcommand{\GRE}{\textsc{Greedy}}
\newcommand{\SGRE}{\textsc{SignedGreedy}}
\newcommand{\BOX}{\textsc{Box}}

\DeclareRobustCommand{\AMmc}{\textsc{F}}
\newcommand{\GG}{\mathcal{G}}
\newcommand{\XX}{\mathcal{X}}
\newcommand{\YY}{\mathcal{Y}}

\newcommand{\RG}{\textsc{Reg}}
\newcommand{\MAXT}{\textsc{MaxTime}}
\newcommand{\MINT}{\textsc{MinTime}}
\newcommand{\MAXK}{\textsc{MaxKey}}
\newcommand{\MINK}{\textsc{MinKey}}
\newcommand{\Left}{\textsc{Left}}
\newcommand{\Right}{\textsc{Right}}
\newcommand{\LeftRel}{\textsc{Left-Rel}}
\newcommand{\RightRel}{\textsc{Right-Rel}}

\newcommand{\PT}{\textsc{Partition}}
\newcommand{\BB}{\mathcal{B}}

\DeclareRobustCommand{\neswarrow}{%
  {\mathrel{\text{\ooalign{$\swarrow$\cr$\nearrow$}}}}%
}

\DeclareRobustCommand{\swnearrow}{%
  {\mathrel{\text{\ooalign{$\swarrow$\cr$\nearrow$}}}}%
}

\DeclareRobustCommand{\senwarrow}{%
  {\mathrel{\text{\ooalign{$\searrow$\cr$\nwarrow$}}}}%
}

\DeclareRobustCommand{\nwsearrow}{%
  {\mathrel{\text{\ooalign{$\searrow$\cr$\nwarrow$}}}}%
}

\DeclareRobustCommand{\abarrow}[2]{
	\overset{#1}{\underset{#2}{\updownarrow}}
}

\DeclareRobustCommand{\lrarrow}[2]{
	#1\!\leftrightarrow\!#2
}

\date{}                                           

\begin{document}

\title{Better Analysis of $\GRE$ Binary Search Tree on Decomposable Sequences}

\author{Navin Goyal\\
   Microsoft Research\\
   Bangalore, India\\
  {\small\texttt{navingo@microsoft.com}}
\and 
Manoj Gupta\\
  IIT Gandhinagar\\
   Gandhinagar, India\\
  {\small\texttt{gmanoj@iitgn.ac.in}}
}  
\maketitle

\begin{abstract}
In their seminal paper [Sleator and Tarjan, J.ACM, 1985], the authors conjectured that 
the {\em splay tree} is {\em dynamically optimal} binary search tree (BST). 
In spite of decades of intensive research, 
the problem remains open. Perhaps a more basic question, which has also attracted much attention, 
is if there exists \emph{any} dynamically 
optimal BST algorithm. 
One such candidate is $\GRE$ which is a simple and intuitive BST algorithm 
[Lucas, Rutgers Tech. Report, 1988; Munro, ESA, 2000; Demaine,
Harmon, Iacono, Kane and Patrascu, SODA, 2009].
[Demaine et al., SODA, 2009] showed a novel connection 
between a geometric problem and the binary search tree problem related to the above conjecture.
However, there has been little progress in solving this geometric problem too.

Since dynamic optimality conjecture in its most general form remains elusive despite much effort, 
researchers have studied this problem on special sequences. Recently, 
[Chalermsook, Goswami, Kozma, Mehlhorn and Saranurak, FOCS, 2015] studied a type of sequences known 
as $k$-{\em decomposable sequences} in this context, where $k$ parametrizes easiness of the sequence. 
Using tools from forbidden submatrix theory, they showed that 
$\GRE$ takes $n2^{O(k^2)}$ time on this sequence and explicitly raised the question of improving this bound. 

In this paper, we show that $\GRE$ takes $O(n \log{k})$ time on $k$-decomposable sequences. In contrast to
the previous approach, ours is based on first principles. 
One of the main ingredients of our
result is a new construction of a lower bound certificate on the performance of any algorithm.
This certificate is constructed using the execution of $\GRE$, and is more nuanced and possibly 
more flexible than the previous independent set certificate of Demaine et al. 
This result, which is applicable to all sequences, 
may be of independent interest and may lead to further progress in analyzing $\GRE$ 
on $k$-decomposable as well as general sequences.
\end{abstract}

\section{Introduction}
\label{sec:intro}
Binary search trees (BSTs) are a well-studied and fundamental data access model. 
We store keys from the universe $\{1,\ldots, n\}$ in a binary search tree and given an sequence of keys
$(p_1, \ldots, p_n)$ we would like to access these keys (and possibly any associated data) using the 
tree. We would like to minimize the total cost of accessing the keys in this sequence, where the cost
for one key search is the number of nodes touched for accessing that key. Various versions of this problem
have been studied and some of them are very well understood, e.g., when the sequence of keys is generated
according to some probability distribution with known access probabilities 
(and some additional restrictions) then optimal search trees are known (see the references in 
\cite{SleatorT85}). But what happens for general access sequences? 
The tree need not be static and can adapt to the access sequence.
Such self-adjusting BST algorithms were considered by Sleator and Tarjan in their seminal paper \cite{SleatorT85}, where they 
introduced splay trees. These trees change dynamically via \emph{rotations} after processing each 
access request.  Sleator and Tarjan showed that the amortized cost of the splay tree is $O(\log n)$.
They famously conjectured that splay trees have the much stronger and attractive property of \emph{dynamic
optimality}:
for any (sufficiently long) sequence the total cost of the splay tree algorithm is within a constant factor of the total cost of 
the optimum \emph{offline} binary search tree (i.e. the cost of the best
BST algorithm that is given the access sequence in advance, and thus can decide how to change the tree
based on this knowledge). Splay trees, by contrast, are \emph{online}: the decision of how to 
change the trees must be based on the current access request only and cannot depend on the future requests. 
The above conjecture is called the {\em dynamic optimality conjecture}.

  
There has been much work on this conjecture, 
as well as on the more fundamental question of whether there 
exists \emph{any} dynamically optimal BST algorithm (see \cite{Iacono13} for a recent review). 
In the past decade progress was made on this latter question and 
BST algorithms with better competitive ratio were discovered:  
Tango trees~\cite{DemaineHIP07} was the first $O(\log \log n)$-competitive BST; 
Multi-Splay trees~\cite{WangDS06} and Zipper trees~\cite{BoseDDF10} also 
have the same competitive ratio along with some additional properties. 
Analyses of these trees use lower bounds for the total cost to process a request sequence.
Wilber~\cite{Wilber89} gave two different lower bounds. 
Wilber's first lower bound is used in \cite{DemaineHIP07, WangDS06, BoseDDF10} 
to obtain $O(\log \log n)$-competitive ratio for the respective trees. 
These techniques 
based on Wilber's bound have so far failed to give $o(\log \log n)$-competitiveness.
Researchers have also proved  
conjectures implied by dynamic optimality; some of these can be interpreted as pertaining 
to {\em easy} sequences, e.g. for splay trees
the amortized access cost is logarithmic in the distance in the key space to the previous key accessed
(dynamic finger theorem) \cite{Cole00,ColeMSS00}, and amortized  accesses cost is logarithmic 
in the temporal distance to the previous access to the current key
(working set theorem) \cite{SleatorT85}.

It turns out that even the problem of designing offline optimal BST algorithm has seen limited progress.
Lucas~\cite{Lucas88} and Munro~\cite{Munro00}  
designed a simple offline greedy BST algorithm and conjectured its cost to be close to the
cost of optimum offline BST algorithm. 
Demaine et al.~\cite{DemaineHIP07}, using a novel geometric point of view of the problem,
showed that surprisingly the offline greedy algorithm can in fact be turned into an online algorithm
called $\GRE$ with only a constant factor loss in the competitive ratio (the ratio between the cost
of the online algorithm to the cost of the offline optimum for the worst case access sequence). 
Thus the conjecture about the the offline greedy would imply that $\GRE$ is dynamically optimal. 
The geometric view of the BST problem mentioned above leads to a completely different looking clean
problem about point sets in the plane. This raises the possibility
of new lines of attack that might be harder to conceive in the BST view. 
Unfortunately, so far success has been limited even with the geometric view. 
The current state of the art by Fox~\cite{Fox11} shows that $\GRE$ takes $O(n \log n)$ time
for any arbitrary sequence $\XX$.

Given this state of affairs, it has been suggested by several researchers to study the problem
for easy sequences (see \cite{ChalermsookG0MS15b} and references therein). 
Just as for splay trees, the question arises about the performance of $\GRE$ on  
{\em easy} sequences; in particular, does the geometric approach help? 
In this context, Chalermsook~et~al.
\cite{ChalermsookG0MS15b} initiated the study of $\GRE$ on {\em decomposable sequences}
and brought to bear techniques from \emph{forbidden submatrix theory} to this problem and some other problems.
(Some of the tools from forbidden submatrix theory have been used previously by 
Pettie~\cite{Pettie08,Pettie10a}
to achieve better bounds for splay trees on deque sequences and a new proof of the sequential access theorem 
for splay trees.) 
They showed that $\GRE$ takes $n 2^{O(k^2)}$ time on $k$-decomposable sequences.
Furthermore, they showed optimal offline cost for $k$-decomposable 
sequences is $\Theta(n \log k)$.  We quote from their paper:

``A question directly related to our work is to close the gap between $\OPT = O(n \log k)$ and $n 2^{O(k^2)}$
by $\GRE$ on $k$-decomposable sequences (when $k = \omega(1)$). Matching the optimum (if at all possible) likely
requires novel techniques: splay is not even known to be linear on preorder sequences with preprocessing,
and with forbidden-submatrix-arguments it seems impossible to obtain bounds beyond $O(nk)$." 

Though the authors mention that forbidden submatrix theory may give an $O(nk)$ bound,
it is not clear how to achieve this goal. 
We solve the above open problem 
by showing 
\begin{theorem} \label{thm:main}
$\GRE$ takes $O(n \log{k})$ time on $k$-decomposable sequences 
{\em (}with preprocessing\footnote{The preprocessing step is the 
same as in \cite{DemaineHIP07,ChalermsookG0MS15b}, that is, 
insert all elements into a {\em split tree} before processing any requests. Preprocessing is independent of the access sequence. See \cite{DemaineHIP07,ChalermsookG0MS15b} for more details}{\em )}. 
\end{theorem}

Our approach is based on first principles and does not use tools from forbidden submatrix theory.
We carefully analyze execution of $\GRE$ on $k$-decomposable sequences and
discover some new structural properties of $\GRE$. 
Our proof also uses the aforementioned general technique of constructing lower bound certificates 
on the cost of the optimum and relating
this lower bound to the cost of the algorithm being analyzed. One such lower bound, 
\emph{independent set} lower bound was provided by Demaine et al.~\cite{DemaineHIP07}. It was, however,
not clear how to relate it, or its close relatives, to the cost of $\GRE$. Our lower bound certificate is 
derived from the execution 
of $\GRE$; it builds upon the ideas of independent set lower bound, but provides a 
more nuanced and possibly more flexible certificate. Our certificate construction works for general sequences and
not just for $k$-decomposable sequences. 
We are hopeful that our approach will lead to further progress in understanding the performance of 
$\GRE$ on $k$-decomposable and general sequences. 


\section{The Geometric Problem}
\label{sec:problem}
Let $[n] = \{1,2, \dots, n\}$ denote the set of keys.
Let $(p_1, p_2, \dots, p_n)$ denote a permutation on $n$ keys, i.e.,
$p_i \neq p_j$ for $i \neq j$.
We can represent this permutation by a set of $n$ points in the plane: 
Key $p_i$ is represented by the point $(p_i,i)$.
For a point $p$, let $\Key{p}$ denote its 
$x$-coordinate and let $\Time{p}$ denote its $y$-coordinate. We will denote sets of points
obtained from permutations of keys in this way by $\XX$. 
There is exactly one point of $\XX$ on the line $x=i$, for $i \in [n]$; 
and there 
is exactly one point from $\XX$ on the line $y=i$, for $i \in [n]$. Clearly, there is a 
one-to-one correspondence between the permutations and sets of points as defined above. 
For the most part we will use the point set view.
 
In our paper, the positive $x$-axis (representing key space) moves from left to right and the 
positive $y$-axis (representing time) moves 
from \emph{top to bottom} (this latter convention is different from most previous papers in this area). 
 For a 
pair of points $p$ and $q$ not on the same horizontal or vertical line, the (closed)
axis-aligned rectangle formed by $p$ and $q$ is denoted by 
$^q\Box_p$ if  $\Time{q} < \Time{p}$ and $\Key{q} < \Key{p}$
and $_p\Box^q$ if  $\Time{q} < \Time{p}$ and $\Key{p} < \Key{q}$.
\iflong 
\begin{center}
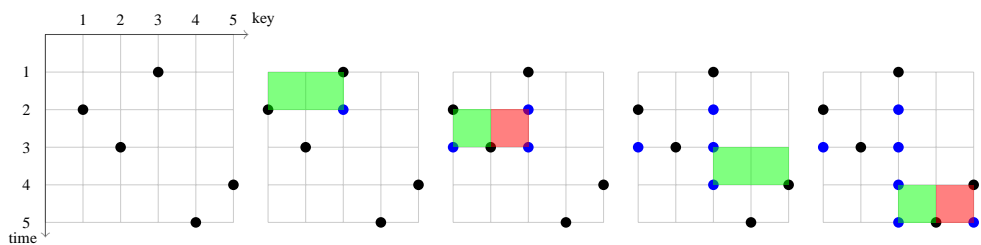
\begin{figure}[H]

\centering
\begin{tikzpicture}

\draw[step=0.5,black,opacity=0.2] (-0.5,0) grid (2,2.5);
\draw [->,opacity=0.4] (-0.5,2.5)--(2.2,2.5);
\draw [->,opacity=0.4] (-0.5,2.5)--(-0.5,-0.2);
\fill [black] (1,2) circle (2pt);
\fill [ black] (0,1.5) circle (2pt);
\fill [ black] (0.5,1) circle (2pt);
\fill [ black] (2,0.5) circle (2pt);
\fill [ black] (1.5,0) circle (2pt);
\draw node at (2.4,2.7) {\tiny{key}};
\draw node at (-0.8,-0.2) {\tiny{time}};
\draw node at (0,2.7) {\tiny{1}};
\draw node at (0.5,2.7) {\tiny{2}};
\draw node at (1,2.7) {\tiny{3}};
\draw node at (1.5,2.7) {\tiny{4}};
\draw node at (2,2.7) {\tiny{5}};

\draw node at (-0.75,2) {\tiny{1}};
\draw node at (-0.75,1.5) {\tiny{2}};
\draw node at (-0.75,1) {\tiny{3}};
\draw node at (-0.75,0.5) {\tiny{4}};
\draw node at (-0.75,0) {\tiny{5}};
\begin{scope}[xshift = 70]
\draw[step=0.5,black,opacity=0.2] (0,0) grid (2,2);

\fill [black] (1,2) circle (2pt);
\fill [ black] (0,1.5) circle (2pt);
\fill [ black] (0.5,1) circle (2pt);
\fill [ black] (2,0.5) circle (2pt);
\fill [ black] (1.5,0) circle (2pt);
\fill [ blue] (1,1.5) circle (2pt);
\fill [green, opacity=0.5] (0,1.5) rectangle (1,2);
\end{scope}

\begin{scope}[xshift = 140]
\draw[step=0.5,black,opacity=0.2] (0,0) grid (2,2);
\fill [black] (1,2) circle (2pt);
\fill [ black] (0,1.5) circle (2pt);
\fill [ black] (0.5,1) circle (2pt);
\fill [ black] (2,0.5) circle (2pt);
\fill [ black] (1.5,0) circle (2pt);
\fill [ blue] (1,1.5) circle (2pt);
\fill [ blue] (1,1) circle (2pt);
\fill [ blue] (0,1) circle (2pt);

\fill [green, opacity=0.5] (0.5,1) rectangle (0,1.5);
\fill [red, opacity=0.5] (0.5,1) rectangle (1,1.5);

\end{scope}

\begin{scope}[xshift = 210]
\draw[step=0.5,black,opacity=0.2] (0,0) grid (2,2);
\fill [black] (1,2) circle (2pt);
\fill [ black] (0,1.5) circle (2pt);
\fill [ black] (0.5,1) circle (2pt);
\fill [ black] (2,0.5) circle (2pt);
\fill [ black] (1.5,0) circle (2pt);
\fill [ blue] (1,1.5) circle (2pt);
\fill [ blue] (1,1) circle (2pt);
\fill [ blue] (0,1) circle (2pt);
\fill [ blue] (1,0.5) circle (2pt);
\fill [green, opacity=0.5] (2,0.5) rectangle (1,1);
\end{scope}

\begin{scope}[xshift = 280]
\draw[step=0.5,black,opacity=0.2] (0,0) grid (2,2);
\fill [black] (1,2) circle (2pt);
\fill [ black] (0,1.5) circle (2pt);
\fill [ black] (0.5,1) circle (2pt);
\fill [ black] (2,0.5) circle (2pt);
\fill [ black] (1.5,0) circle (2pt);
\fill [ blue] (1,1.5) circle (2pt);
\fill [ blue] (1,1) circle (2pt);
\fill [ blue] (0,1) circle (2pt);
\fill [ blue] (1,0.5) circle (2pt);
\fill [ blue] (1,0) circle (2pt);
\fill [ blue] (2,0) circle (2pt);
\fill [green, opacity=0.5] (1.5,0) rectangle (1,0.5);
\fill [red, opacity=0.5] (1.5,0) rectangle (2,0.5);
\end{scope}
\end{tikzpicture}
\caption{The first picture  shows the point set $X$=\{3,1,2,5,4\}. In the remainder of the paper 
we do not show the x-axis and the y-axis (along with the first row and column of the grid).  
The execution of $\GRE$ is shown from the second picture onwards.}
\label{fig:greedy-execution}
\end{figure}
\end{center}
\fi 


\begin{definition}
\label{def:rectangleDef}
A pair of points $(p,q)$ is said  to be arborally satisfied with respect to 
a point set $P$, if (1) $p$ and $q$ lie on the same horizontal or vertical line or,
(2) $\exists r \in P \setminus \{p,q\}$ such that $r$ lies in the interior  
or boundary of $^q\Box_p$ {\em(}or $_p\Box^q${\em )}.
\end{definition}


We say that the rectangle $^q\Box_p$ (or $_p\Box^q$) is arborally satisfied if condition (2) holds, otherwise
it is arborally unsatisfied.
Consider the following problem: \emph{Given a point set $\XX$, find a minimum 
cardinality  point set $\YY$ such that each pair of point in the set 
$\XX \cup \YY$ is arborally satisfied.}

If each pair of points in $\XX \cup \YY$ is arborally satisfied then we say that the set 
$\XX \cup \YY$ is an arborally satisfied set otherwise it is not.
In their remarkable paper, Demaine~et~al.~\cite{DemianeHIKP09} formulated the above problem.
They showed a novel connection between this {\em geometric} problem 
and binary search tree (BST) problem, and designed a simple algorithm, henceforth 
called $\GRE$, for the above geometric problem. Let $\GG$ be the set of points added by
this algorithm described as follows: 

\emph{ Sweep the point set $\XX$ with a horizontal line by increasing the y-coordinates. 
Let the point $p$ be processed at time $\Time{p}$. At time $\Time{p}$, place the minimal 
number of points on line $y=\Time{p}$ to satisfy the rectangle with $p$ as one 
endpoint and other endpoint in $\XX \cup \GG$ having their y-coordinate less than 
$\Time{p}$. This minimal set of points $M_{p}$ is uniquely defined:
for any arborally unsatisfied rectangle formed with
$p$ in one corner, add a point at the other corner at $y = \Time{p}$ in $\GG$.
Please see Figure \ref{fig:greedy-execution} for 
the execution of $\GRE$.\\}
One can show that the $\XX \cup \GG$ is an arborally satisfied set; see \cite{DemianeHIKP09} for details.

\section{Overview}
\label{sec:overview}
We give a brief overview of our techniques in this section.
Our goal is to prove $|\GG| = O(n \log{k})$ which immediately implies Theorem~\ref{thm:main}. 
The starting point of our approach was an attempt to construct an independent set of rectangles of original points
certifying a lower bound on the number of points that must be marked. We attempt to construct such a certificate
by analyzing the execution of $\GRE$. Our final certificate will not be an independent set however. 
 
In Sec.~\ref{sec:coupling}, for each point in $\GG$ we associate a tuple of points 
(which we also think of as a rectangle) from $\XX$ using a map 
called $\CP(\cdot)$. At a high level,
this can be thought of looking for a reason for why the point in $\GG$ was marked by $\GRE$. 
We partition $\CP(\GG)$ into two sets called $\MMC$ (pronounced zig) and $\MFC$ (pronounced zag). The visual notation
$\MMC$ and $\MFC$ depicts how $\CP(p)$ is located w.r.t. $p \in \GG$. 
In Sections \ref{sec:mmc}  we show that $|\MMC| = O(nk)$ and in Sec.~\ref{sec:mmc-betterbound} we improve it to 
$|\MMC| = O(n \log{k})$; this has a relatively short proof and uses properties
of $\MMC$ and $k$-decomposable sequences.

We then show $|\MFC| = O(nk)$. The proof of this is in two parts. First we construct the partition 
$\MFC = \GR \cup \BR$ (Sec.~\ref{subsubsec:properties}). 
The set $\GR$ consists of rectangles from $\MFC$ that do not have any original points in their interior (thus this
set can presumably be quite different from an independent set).
In Sections~\ref{subsubsec:interaction} and \ref{subsubsec:finaldetail}, we analyze $\GR$ and show that it provides
a lower bound certificate for $|\OPT(\XX)|$ similar to the independent set certificate of 
Demaine~et~al.~\cite{DemianeHIKP09}: 
$|\GR|/2 + |\XX| \le |\OPT(\XX) \cup \XX|$. We remark that this result 
holds for all 
$\XX$ and not just for $k$-decomposable sequences, and hence may be of use in future work on the general problem.
Chalermsook~et~al. have
proven $|\XX \cup \OPT(\XX)| \le  O(n \log k)$ for $k$-decomposable sequences, 
which implies $|\GR| = O(n \log k)$. 
Finally, in Sec.~\ref{sec:mfc} we show $|\BR| = O(nk)$ and then in Sec.~\ref{sec:mfc-better} improve it to 
$|\BR| = O(n \log{k})$.
For this, we use some structural properties of $\BR$, $\MMC$, and $k$-decomposable sequences. 
The above results together imply our desired bound $|\GG| = |\CP(\GG)| = O(n \log k)$ (Sec.~\ref{sec:main}).

\section{Basic Properties of $\GRE$}
\label{sec:prelims}
In all our diagrams, a point in $\XX$ is denoted by a black circle and
a point in $\GG$ is denoted by a blue circle. A point in $(\XX \cup \GG)$ 
is denoted by a gray circle. A point in $\XX$ will be called an \emph{original} point 
and a point in $\GG$ will be called a \emph{marked} point. 
When we refer to a point without specifying whether its marked or original, then it is a point 
from $(\XX \cup \GG)$. 
We use following notation:
\begin{enumerate}
\item $p$ is \emph{above} $q$ (or $q$ is below $p$) or $\abarrow{p}{q}$
 if $\Time{p} < \Time{q}$ and $\Key{p} = \Key{q}$.

\item $p$ is to the \emph{left} of $q$ (or $q$ is to the right of $p$) or $\lrarrow{p}{q}$ if $\Key{p} < \Key{q}$ and $\Time{p} = \Time{q}$.

\item $q$ is to the \emph{south-east} of $p$ (or $p$ is to the north-west of $q$) or $^p{\senwarrow}_q$ if 
$\Time{p} < \Time{q}$ and $\Key{p} < \Key{q}$.

\item $q$ is to the \emph{north-east} of $p$ (or $p$ is to the south-west of $q$) or $_p{\neswarrow}^q$ if 
$\Time{p} > \Time{q}$ and $\Key{p} < \Key{q}$.

\item $(_p\Box^q)^{\circ}$ denotes the interior of $_p\Box^q$.
	

\item While processing $p \in \XX$, $\GRE$ may put marked points on the line $y = \Time{p}$ (there is no other
original point on this line as $\XX$ comes from a permutation sequence). 
For any such marked point $q$ we define its original point $\OP(q)$ to be $p$. We also set $\OP(p) := p$.

\iflong
\item For $q \in \XX \cup \GG$, define $\UP(q):=q$ if $q \in \XX$, and 
$\UP(q):=p$ if $q \in \GG$, where $p \in \XX$ is the unique original point above $q$ (see the discussion 
before Observation~\ref{obs:nothing-above} below).
\fi
\end{enumerate}

\iflong 
	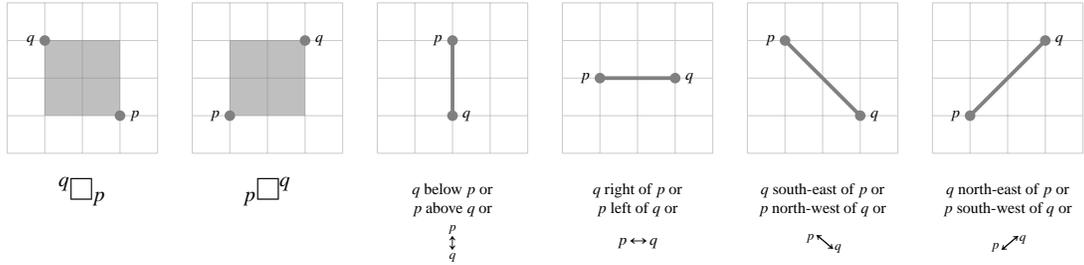
\begin{figure}[H]

\centering
\begin{tikzpicture}

\draw[step=0.5,black,opacity=0.2] (0,0) grid (2,2);

\fill [gray] (1.5,0.5) circle (2pt);

\fill [ gray] (0.5,1.5) circle (2pt);
\node at (0.5,1.5) [left] {\tiny{$q$}};
\node at (1.5,0.5) [right] {\tiny{$p$}};
\fill [gray, opacity=0.5] (0.5,0.5) rectangle (1.5,1.5);
\node at (1,-0.5) {$^q\Box_p$};

\begin{scope}[xshift = 70]
\draw[step=0.5,black,opacity=0.2] (0,0) grid (2,2);

\fill [gray] (0.5,0.5) circle (2pt);
\fill [ gray] (1.5,1.5) circle (2pt);

\node at (0.5,0.5) [left] {\tiny{$p$}};
\node at (1.5,1.5) [right] {\tiny{$q$}};

\fill [gray, opacity=0.5] (0.5,0.5) rectangle (1.5,1.5);
\node at (1,-0.5) {$_p\Box^q$};
\end{scope}

\begin{scope}[xshift = 140]
\draw[step=0.5,black,opacity=0.2] (0,0) grid (2,2);
\fill [gray] (1,1.5) circle (2pt);
\fill [ gray] (1,0.5) circle (2pt);

\node at (1,1.5) [left] {\tiny{$p$}};
\node at (1,0.5) [right] {\tiny{$q$}};
\draw [gray,line width=0.5mm] (1,1.5) -- (1,0.5); 
\node at (1,-0.5) {\tiny{$q$ below $p$ or}};
\node at (1,-0.75) {\tiny{$p$ above $q$ or }};
\node at (1,-1.2) {\tiny{$\abarrow{p}{q}$}};
\end{scope}

\begin{scope}[xshift = 210]
\draw[step=0.5,black,opacity=0.2] (0,0) grid (2,2);
\fill [gray] (0.5,1) circle (2pt);
\fill [ gray] (1.5,1) circle (2pt);

\node at (0.5,1) [left] {\tiny{$p$}};
\node at (1.5,1) [right] {\tiny{$q$}};
\draw [gray,line width=0.5mm] (0.5,1) -- (1.5,1); 
\node at (1,-0.5) {\tiny{$q$ right of $p$ or}};
\node at (1,-0.75) {\tiny{$p$ left of $q$ or }};
\node at (1,-1.2) {\tiny{$\lrarrow{p}{q}$}};
\end{scope}

\begin{scope}[xshift = 280]
\draw[step=0.5,black,opacity=0.2] (0,0) grid (2,2);
\fill [gray] (0.5,1.5) circle (2pt);
\fill [ gray] (1.5,0.5) circle (2pt);

\node at (0.5,1.5) [left] {\tiny{$p$}};
\node at (1.5,0.5) [right] {\tiny{$q$}};
\draw [gray,line width=0.5mm] (0.5,1.5) -- (1.5,0.5); 
\node at (1,-0.5) {\tiny{$q$ south-east of $p$ or}};
\node at (1,-0.75) {\tiny{$p$ north-west of $q$ or }};
\node at (1,-1.2) {\tiny{ $^p{\senwarrow}_q$}};
\end{scope}

\begin{scope}[xshift = 350]
\draw[step=0.5,black,opacity=0.2] (0,0) grid (2,2);
\fill [gray] (0.5,0.5) circle (2pt);
\fill [ gray] (1.5,1.5) circle (2pt);

\node at (0.5,0.5) [left] {\tiny{$p$}};
\node at (1.5,1.5) [right] {\tiny{$q$}};
\draw [gray,line width=0.5mm] (0.5,0.5) -- (1.5,1.5); 
\node at (1,-0.5) {\tiny{$q$ north-east of $p$ or}};
\node at (1,-0.75) {\tiny{$p$ south-west of $q$ or }};
\node at (1,-1.2) {\tiny{ $_p{\neswarrow}^q$ }};
\end{scope}

\end{tikzpicture}
\caption{Basic Notations}
\label{fig:notation}
\end{figure}
	We now show some basic properties and lemmas related to the execution of $\GRE$. 

	While preprocessing $p$, $\GRE$ adds a marked point at the bottom right (bottom left) corner of rectangle 
	$_p\Box^q$ $(^q\Box_p)$ only if it is an arborally unsatisfied rectangle.
	This implies that whenever $\GRE$ puts a marked point there exists another point (marked or original)
	above it. Using this property, we claim that the top point on the line $x=i$ ($1\le i \le n$) must
	be an original point, i.e., a point from $\XX$. The following observation follows: 

	\begin{observation}
	\label{obs:nothing-above}
	For any point $p \in \XX$, $\GRE$ does not put any marked point above $p$.
	\end{observation}

\fi
We now prove some lemmas regarding the execution of $\GRE$:

\iflong
	\begin{lemma}
	\label{lem:greedy-property}
\else
	\newtheorem*{lem:greedy-property}{Lemma \ref{lem:greedy-property}}
	\begin{lem:greedy-property}
\fi
Consider a rectangle $_p\Box^q$ where $p,q \in (\XX \cup \GG)$. 
Then there exists a point $r \in (\XX \cup \GG) \setminus \{p,q\}$ such that
(1) $r \in$ $_p\Box^q$, and
(2) $\abarrow{r}{p}$ or $\lrarrow{p}{r}$.

\iflong
	\end{lemma}
\else	
	\end{lem:greedy-property}
\fi
\iflong 
	\begin{remark}
	This lemma is a special case of Observation 2.1 in \cite{DemianeHIKP09}. The lemma
	is true for any point set $\YY$ for which $(\XX \cup \YY)$ is an arborally 
	satisfied set.
	\end{remark}
	\begin{figure}[H]

\centering
\begin{tikzpicture}

\draw[step=0.25,black,opacity=0.2] (0,0) grid (2,2);

\fill [gray , opacity=0.7] (0.5,0.5) circle (2pt);
\fill [ gray , opacity=0.7] (1.5,1.5) circle (2pt);
\node at (0.5,0.5) [left] {\tiny{$p$}};
\node at (1.5,1.5) [right] {\tiny{$q$}};

\fill [gray, opacity=0.2] (0.5,0.5) rectangle (1.5,1.5);

\draw [gray,line width=0.5mm] (0.5,0.5) -- (0.5,1.5); 
\draw [gray,line width=0.5mm] (0.5,0.5) -- (1.5,0.5); 

\node at (1,-0.5) {(i)};

\begin{scope}[xshift = 70]
\draw[step=0.25,black,opacity=0.2] (0,0) grid (2,2);

\fill [gray , opacity=0.7] (0.5,0.5) circle (2pt);
\fill [ gray , opacity=0.7] (1.5,1.5) circle (2pt);
\node at (0.5,0.5) [left] {\tiny{$p$}};
\node at (1.5,1.5) [right] {\tiny{$q$}};

\fill [gray, opacity=0.2] (0.5,0.5) rectangle (1.5,1.5);

\fill [gray , opacity=0.7] (0.75,0.75) circle (2pt) node[right,opacity=1,black] {\tiny{$s$}};
\node at (1,-0.5) {(ii)};

\end{scope}

\end{tikzpicture}
\caption{(1) Lemma \ref{lem:greedy-property} states that there exists a point
in $(\XX \cup \GG)\setminus \{p,q\}$ on the thick lines adjacent to $p$. 
(2) Illustration of the case argued in the proof where $s$ is the closest point 
to $p$ in $_p\Box^q$.}
\label{fig:greedy-peoperty}
\end{figure}
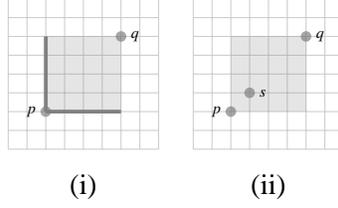
	\begin{proof}
	Since $(\XX \cup \GG)$ is an arborally satisfied set, there exists another point from $(\XX \cup \GG)\setminus \{p,q\}$ 
	in $_p\Box^q$.
	Let $s$ be the closest point to $p$ in $_p\Box^q$ (closest in Euclidean distance) (See Figure \ref{fig:greedy-peoperty}(ii)). 
	If $\abarrow{s}{p}$ or $\lrarrow{p}{s}$, then we are done. 
	Else, look at the rectangle $_p\Box^s$. Since $s$ is the closest point to $p$,
	$_p\Box^s$ does not contain any other point from $(\XX \cup \GG) \setminus \{p,s\}$. 
	This implies that $_p\Box^s$ is not arborally satisfied.
	This leads to a contradiction as $(\XX \cup \GG)$ is an arborally satisfied set.
	\end{proof}
\fi

\iflong
	Similarly, we can also prove a symmetric version of the above lemma:
	\begin{lemma}
	\label{lem:greedy-property-left}
	Consider a rectangle $^q\Box_p$ where $p,q \in (\XX \cup \GG)$. 
	Then there exists a point $r \in (\XX \cup \GG) \setminus \{p,q\}$ such that
	(1) $r \in$ $^q{\Box}_p$, and
	(2) $\abarrow{r}{p}$ or $\lrarrow{r}{p}$.
	\end{lemma}
\fi

We now move on to another important property of $\GRE$.
\iflong
	\begin{definition}
	\label{def:hides}
\else
	\newtheorem*{def:hides}{Definition \ref{def:hides}}
	\begin{def:hides}
\fi
A point $r$ \emph{hides} $q$ from $p$ where $p,q,r \in \XX \cup \GG$ and
$^q\nwsearrow_p$, if either
(1) $\lrarrow{q}{r}$ 
  and $r \in$ $^q\Box_p$, or 
(2) $r \in (^q\Box_p)^{\circ}$. 

For $p, q$ such that $_p\neswarrow^q$ the definition is symmetric. 

\iflong
	\end{definition}
\else	
	\end{def:hides}
\fi
In other words, $r$ hides $q$ from $p$, if it's different from $p, q$ and 
lies in the union of the ``top-line'' and the interior of $^q\Box_p$ (or $_p\Box^q$). 
\iflong 
	\begin{figure}[H]

\centering
\begin{tikzpicture}

\draw[step=0.25,black,opacity=0.2] (0,0) grid (2,2);

\fill [gray] (1.5,0.5) circle (2pt);

\fill [ gray, opacity=0.7] (0.5,1.5) circle (2pt);
\node at (0.5,1.5) [left] {\tiny{$q$}};
\node at (1.5,0.5) [right] {\tiny{$p$}};
\fill [gray, opacity=0.2] (0.5,0.5) rectangle (1.5,1.5);

\fill [gray , opacity=0.7] (0.75,1.25) circle (2pt) node[right,opacity=1,black] {\tiny{$r$}};


\begin{scope}[xshift = 70]
\draw[step=0.25,black,opacity=0.2] (0,0) grid (2,2);

\fill [gray] (0.5,0.5) circle (2pt);
\fill [ gray , opacity=0.7] (1.5,1.5) circle (2pt);

\fill [gray , opacity=0.7] (1.25,1.5) circle (2pt) node[above,opacity=1,black] {\tiny{$r$}};
\node at (0.5,0.5) [left] {\tiny{$p$}};
\node at (1.5,1.5) [right] {\tiny{$q$}};

\fill [gray, opacity=0.2] (0.5,0.5) rectangle (1.5,1.5);
\end{scope}

\end{tikzpicture}
\caption{$r$ hides $q$ from $p \in \XX$}
\label{fig:hides}
\end{figure}
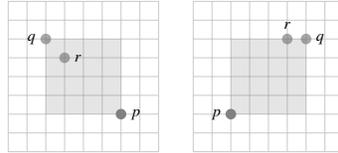
\fi

\iflong
	\begin{lemma}
	\label{lem:hidden}
\else
	\newtheorem*{lem:hidden}{Lemma \ref{lem:hidden}}
	\begin{lem:hidden}
\fi

Let $r$ hide $q$ from $p$ where $p,q,r \in \XX \cup \GG$ and $^q\nwsearrow_p$. 
Assume that there exists a point $s$ below $q$ such that $_{s}\swnearrow^{r}$. Moreover, let $s$ be the first 
point below $q$ with this property. 
Then $_{\OP(s)}\swnearrow^r$.
\iflong
	\end{lemma}
\else	
	\end{lem:hidden}
\fi

\iflong 
	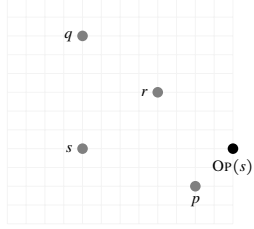
\begin{figure}[H]

\centering
\begin{tikzpicture}
\draw[step=0.25,black,opacity=0.05] (0,0) grid (3,3);

\fill [ gray] (2.5,0.5) circle (2pt) node[below,black,opacity=1] {\tiny{$p$}};

\fill [ gray] (1,2.5) circle (2pt) node[left,black,opacity=1] {\tiny{$q$}};
\fill [ gray] (2,1.75) circle (2pt) node[left,black,opacity=1] {\tiny{$r$}};
\fill [ black] (3,1) circle (2pt) node[below,black,opacity=1] {\tiny{$\OP(s)$}};
\fill [ gray] (1,1) circle (2pt) node[left,black,opacity=1] {\tiny{$s$}};




\end{tikzpicture}

\caption{ Illustration of bad case in the proof of Lemma \ref{lem:hidden} when $^r\nwsearrow_{\OP(s)}$}
\label{fig:hidden}
\end{figure}
	\begin{proof}
	Assume that $r$ lies in $(^q\Box_p)^{\circ}$ (the argument below applies even if $\lrarrow{q}{r}$).
	By Observation \ref{obs:nothing-above}, $\GRE$ does not put any marked 
	point above any original point. So $\OP(s)$ cannot lie below $r$.
	Assume then for contradiction that $^r\senwarrow_{\OP(s)}$ (See Figure \ref{fig:hidden}).
	Since $\GRE$ puts $s$ while processing $\OP(s)$, it must have encountered an unsatisfied rectangle
	$^{t}\Box_{\OP(s)}$ such that $t$ is the first point above $s$. 
	We claim that $t$ cannot lie to the north-west or left of $r$ 
	as then $^{t}\Box_{\OP(s)}$ is arborally satisfied by $r$. This implies that $_{t}\swnearrow^r$. 
	But then $t$ is the first point below 
	$q$ such that $_{t}\swnearrow^{r}$, which contradicts the assumption of the lemma. 
	So our assumption that $^r\senwarrow_{\OP(s)}$ must be false.
	\end{proof}
\fi

\iflong
	We also state the symmetric version of the above lemma:
	\begin{lemma}
	\label{lem:hidden-right}
	Let $r$ hide $q$ from $p$ where $p,q,r \in \XX \cup \GG$ and $_p\neswarrow^q$. 
	Assume that there exists a point $s$ 
	below $q$ 
	such that $^r\nwsearrow_{s}$.
	Moreover, let $s$ be the first 
	point below $q$ with this property. 
	Then $^r\nwsearrow_{\OP(s)}$.
	\end{lemma}
\fi

\section{Decomposable Sequences}
\label{subsec:decomposable-sequence}
Given a permutation of the keys $(p_1, p_2, \dots, p_n)$, represented by point set $\XX$ in the plane as described above, we call a set 
$[i, j] := \{i, i+1, \ldots, j\}$ a block of $\XX$
if $\{p_i, p_{i+1}, \dots, p_j\} = \{c, c+1, \dots, d\}$ for some $c, d \in [n]$.
In words, a block represents a contiguous time interval that is mapped to a contiguous 
key interval by the permutation. 
We say that $\XX$
is decomposable into $k$-blocks if there exist disjoint blocks $[a_1,b_1], \dots , [a_k,b_k]$ 
such that
$\cup_\ell [a_\ell,b_\ell] \cap  \mathbb{N} = [n]$. 

We can recursively decompose $\XX$
till singleton blocks are obtained. This recursive decomposition can be naturally 
represented as a rooted tree where each node represents a block. At the root of the tree is the block 
$\XX = (p_1, p_2, \dots, p_n)$. Let us call this tree a \emph{decomposition tree}
of $\XX$. 
We say that $\XX$ is \emph{$k$-decomposable}
if there exists a decomposition tree $\TD(\XX)$ such that the number of children 
of each internal node in this tree is at most $k$.
\iflong 
	\begin{center}
\begin{figure}[H]
\centering
\begin{tikzpicture}
[round/.style={rectangle, rounded corners=3mm, minimum size=20mm, draw=black!50}]
\draw[step=0.5,black,opacity=0.2] (0,0) grid (2,2);

\fill [black] (1,2) circle (2pt);
\fill [ black] (0,1.5) circle (2pt);
\fill [ black] (0.5,1) circle (2pt);
\fill [ black] (2,0.5) circle (2pt);
\fill [ black] (1.5,0) circle (2pt);

\begin{scope}[xshift = 100]
\draw[step=0.5,black,opacity=0.2] (0,0) grid (2,2);

\fill [black] (1,2) circle (2pt);
\fill [ black] (0,1.5) circle (2pt);
\fill [ black] (0.5,1) circle (2pt);
\fill [ black] (2,0.5) circle (2pt);
\fill [ black] (1.5,0) circle (2pt);
\draw [round,fill=gray,opacity=0.2] (-0.55,0.55) rectangle (1.25,2.25);
\draw [round,fill=gray,opacity=0.4] (-0.35,0.765) rectangle (0.85,1.85);
\draw [round,fill=gray,opacity=0.4] (0.7,1.7) rectangle (1.3,2.3);

\draw [round,fill=gray,opacity=0.6](-0.3,1.2) rectangle (0.3,1.8);
\draw [round,fill=gray,opacity=0.6](0.2,0.7) rectangle (0.8,1.3);

\draw [round,fill=gray,opacity=0.2] (1.2,-.35) rectangle (2.3,0.8);

\draw [round,fill=gray,opacity=0.4] (1.2,-.25) rectangle (1.8,0.3);
\draw [round,fill=gray,opacity=0.4] (1.7,.25) rectangle (2.3,0.8);
\end{scope}
\end{tikzpicture}
\caption{The point set $\XX$=\{3,1,2,5,4\} and its decomposition}
\label{fig:decompositionX}
\end{figure}
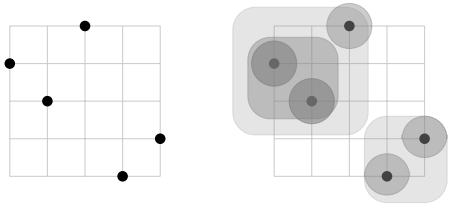

\end{center}


\fi
Let $\PP(B)$ denote the parent of $B$ in $\TD(\XX)$, and let $\SB(B)$ denote the set of children of
$\PP(B)$ except $B$. Let $\Top(B) := a$
if $a \in B$ and $\Time{a} \le \Time{p}$ for all $p \in B$. In words,
$\Top(B)$ denotes the first (in time) key in block $B$. 
Note that a key $a$ can be the first
key of multiple blocks (at different levels). In Fig.~\ref{fig:decompositionX}
\iflong
,
\fi
\Top($3,1,2,5,4) = \Top(3,1,2) = \Top(3) = 3$.
Let $\RT(a)$ be the block $B$ that is closest to the root and satisfies $\Top(B) = a$.
In Fig.~\ref{fig:decompositionX}, $\RT(3) = (3,1,2,5,4)$.

In the rest of this paper we deal with $k$-decomposable permutations $\XX$, or more precisely, with 
point sets $\XX$ representing such permutations. 
We fix some $\TD(\XX)$ such that every internal node has at most $k$ children. 
Henceforth when we talk about blocks, it will be with respect to this fixed $\TD(\XX)$.

\iflong 
	For a block $B$, let $\BOX(B) := \{ r \mid \exists p,q \in B$ s.t. $\Key{r}=\Key{p}$ and $\Time{r}=\Time{q}\}$.
	In words, $\BOX(B)$ contains those points $r$ in the plane such that both the 
	horizontal and vertical lines passing through $r$ have at least one point from $B$.

	Define 
	$\UPB(B) := \{ r~|~r.y < \Time{\Top(B)}$ and $\exists p \in B$ such that $\Key{p} = \Key{r} \}$.
	In words, $\UPB(B)$ is the set of points in the plane that come before all points in $B$ in time, but
	share the key with some point in $B$. 
	Please see the Figure \ref{fig:upperbox}.

	\begin{center}
\begin{figure}[H]
\centering
\begin{tikzpicture}
[round/.style={rectangle, rounded corners=3mm, minimum size=20mm, draw=black!50}]

\draw[step=0.5,black,opacity=0.2] (0,0) grid (2,2);

\fill [black] (1,2) circle (2pt);
\fill [ black] (0,1.5) circle (2pt);
\fill [ black] (0.5,1) circle (2pt);
\fill [ black] (2,0.5) circle (2pt);
\fill [ black] (1.5,0) circle (2pt);

\draw [round,fill=gray,opacity=0.6] (1.4,0.9) rectangle (2.1,2.1);

\draw [round,fill=gray,opacity=0.2] (1.2,-.35) rectangle (2.3,0.8);

\draw node at (1.7,-0.5) {\small{$\BOX(B)$}};
\draw node at (1.7,2.3) {\small{$\UPB(B)$}};

\end{tikzpicture}
\caption{ $\BOX(B)$ and $\UPB(B)$ contains the set of all the points 
in the two region shown in the figure.}
\label{fig:upperbox}
\end{figure}
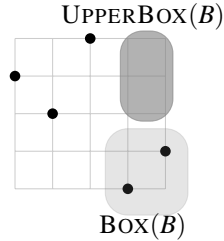

\end{center}

	By definition, a block represents a contiguous key interval.
	Hence, there exists no point $q \in \XX$ in 
	$\UPB(B)$. Also
	by Observation \ref{obs:nothing-above}, for any point $q \in \XX$, 
	GREEDY does not put any point above $q$, so we have

	\begin{lemma}
	\label{lem:upperbox}
	There is no point from $(\XX \cup \GG)$ in $\UPB(B)$.
	\end{lemma}
\fi

\begin{remark}
For brevity, our definitions and theorem and lemma statements do not explicitly quantify over $\XX$, but there is always an underlying \lq\lq For a permutation
sequence $\XX$\rq\rq. In Sections \ref{sec:coupling} and \ref{sec:goodrectangles}, this quantification is over all $\XX$ of size $n$; and in 
Sections \ref{sec:mmc}, \ref{sec:mfc}, and \ref{sec:main}, it is over all $k$-decomposable $\XX$ of size $n$. 
\end{remark}

\section{Pairs}
\label{sec:coupling}
We define a map from $\GG$ to pairs of points in $\XX$.  
\begin{definition}

Let $p \in \GG$. The map $\CP : \GG \rightarrow \XX^2$ is defined as follows:
Let $q$ be a the first point above $p$, then 
$\CP(p) := (\OP(p),\OP(q))$.
\end{definition}

Note that if $\RR(p) = (\OP(p),\OP(q))$, then $\Time{\OP(p)} > \Time{\OP(q)}$.

Our goal is to upper bound $|\GG|$. We do this by connecting $|\GG|$ to the set of all \emph{pairs}
$\CP(\GG) := \{ \CP(p) ~|~ p \in \GG \}$, partitioning the set $\CP(\GG)$ into two sets $\MMC$ and 
$\MFC$, and then bounding $|\MMC|$ and $|\MFC|$ from above.

\begin{definition}[\MMC]
Subset $\MMC$ of $\CP(\GG)$, pronounced \emph{zig}, is defined as follows. 
For $p \in \GG$, let $q$ be the first point above $p$ and $\CP(p) = (\OP(p),\OP(q))$. We say 
that $\CP(p) \in \MMC$ if
(1) $\lrarrow{\OP(p)}{p}$ and $^{\OP(q)}\nwsearrow_{\OP(p)}$, or its symmetric version
(2) $\lrarrow{p}{\OP(p)}$ and $_{\OP(p)}\neswarrow^{\OP(q)}$.
\end{definition}
\begin{figure}[H]
\centering
\begin{tikzpicture}

\draw[step=0.25,black,opacity=0.1] (0,0) grid (3,2);

\fill [black] (1,0.5) circle (2pt) node[below,black,opacity=1] {\tiny{$\OP(p)$}};
\fill [ blue] (2,0.5) circle (2pt) node[below,black,opacity=1] {\tiny{$p$}};

\fill [ blue] (2,1) circle (2pt) node[above,black,opacity=1] {\tiny{$q$}};

\fill [ black] (0.5,1) circle (2pt) node[above,black,opacity=1] {\tiny{$\OP(q)$}};
\node at (1.5,-0.5) {(i)};

\begin{scope}[xshift=100]
\draw[step=0.25,black,opacity=0.1] (0,0) grid (3,2);
\fill [black] (2,0.5) circle (2pt) node[below,black,opacity=1] {\tiny{$\OP(p)$}};
\fill [ blue] (1,0.5) circle (2pt) node[below,black,opacity=1] {\tiny{$p$}};

\fill [ blue] (1,1) circle (2pt) node[above,black,opacity=1] {\tiny{$q$}};

\fill [ black] (2.5,1) circle (2pt) node[above,black,opacity=1] {\tiny{$\OP(q)$}};
\node at (1.5,-0.5) {(ii)};
\end{scope}

\end{tikzpicture}
\caption{$\CP(p) \in \MMC$. In Figure (i) $\CP(p) \in \RMMC$ and in Figure (ii) $\CP(p) \in \LMMC$}
\label{fig-coupling}
\end{figure}
Note that the symbol $\MMC$ mimics Fig.~\ref{fig-coupling}(i).
Let $\RMMC = \{ \CP(p)~|~\CP(p) \in \MMC$ and $\lrarrow{\OP(p)}{p} \}$. 
Similarly, $\LMMC = \{ \CP(p)~|~\CP(p) \in \MMC$ and $\lrarrow{p}{\OP(p)} \}$.

\begin{definition}[\MFC]
Subset $\MFC$ of $\CP(\GG)$, pronounced \emph{zag}, is defined as follows. 
For $p \in \GG$, let $q$ be the first point above $p$ and $\CP(p) = (\OP(p),\OP(q))$. We say 
that $\CP(p) \in \MFC$ if
(1) $\lrarrow{\OP(p)}{p}$ and ($\abarrow{\OP(q)}{p}$ or $_{p}\neswarrow^{\OP(q)}$), or its symmetric version
(2) $\lrarrow{p}{\OP(p)}$ and ($\abarrow{\OP(q)}{p}$ or $^{\OP(q)}\nwsearrow_{p}$).
\end{definition}
\begin{figure}[H]
\centering
\begin{tikzpicture}

\draw[step=0.25,black,opacity=0.1] (0,0) grid (2,1.5);

\fill [black] (0.5,0.5) circle (2pt) node[left,black,opacity=1] {\tiny{$\OP(p)$}};
\fill [ blue] (1.5,0.5) circle (2pt) node[below,black,opacity=1] {\tiny{$p$}};

\fill [ black] (1.5,1) circle (2pt) node[above,black,opacity=1] {\tiny{$q$}};
\node at (1,-0.5) {(i)};

\begin{scope}[xshift=70]
\draw[step=0.25,black,opacity=0.1] (0,0) grid (2,1.5);
\fill [black] (0.5,0.5) circle (2pt) node[left,black,opacity=1] {\tiny{$\OP(p)$}};
\fill [ blue] (1.5,0.5) circle (2pt) node[below,black,opacity=1] {\tiny{$p$}};
\fill [ blue] (1.5,1) circle (2pt) node[above,black,opacity=1] {\tiny{$q$}};
\fill [ black] (2,1) circle (2pt) node[above,black,opacity=1] {\tiny{$\OP(q)$}};
\node at (1,-0.5) {(ii)};
\end{scope}

\begin{scope}[xshift=140]
\draw[step=0.25,black,opacity=0.1] (0,0) grid (2,1.5);
\fill [black] (1.5,0.5) circle (2pt) node[below,black,opacity=1] {\tiny{$\OP(p)$}};
\fill [ blue] (0.5,0.5) circle (2pt) node[below,black,opacity=1] {\tiny{$p$}};

\fill [ black] (0.5,1) circle (2pt) node[above,black,opacity=1] {\tiny{$q$}};
\node at (1,-0.5) {(iii)};
\end{scope}

\begin{scope}[xshift=210]
\draw[step=0.25,black,opacity=0.1] (0,0) grid (2,1.5);
\fill [black] (1.5,0.5) circle (2pt) node[below,black,opacity=1] {\tiny{$\OP(p)$}};
\fill [ blue] (0.5,0.5) circle (2pt) node[below,black,opacity=1] {\tiny{$p$}};

\fill [ blue] (0.5,1) circle (2pt) node[above,black,opacity=1] {\tiny{$q$}};

\fill [ black] (0,1) circle (2pt) node[above,black,opacity=1] {\tiny{$\OP(q)$}};
\node at (1,-0.5) {(iv)};
\end{scope}

\end{tikzpicture}
\caption{$\CP(p) \in \MFC$.}
\label{fig:mfc}
\end{figure}
Note that the symbol $\MFC$ mimics Fig.~\ref{fig:mfc}(ii).

\iflong 
	The following observation follows from the definition of $\MFC$:
	\begin{observation}
	\label{obs:pointinside}
	If $\CP(p) = (\OP(p),\OP(q))$ and $\CP(p) \in \MFC$, then $p$ lies to the 
	left {\em(}right{\em)} of $\OP(p)$ in $^{\OP(q)}\Box_{\OP(p)} (_{\OP(p)}\Box^{\OP(q)})$.
	\end{observation}

	Henceforth, we will
	abuse notation and use $\CP(p_1)$ (with $\CP(p_1) \in \MFC$) as an 
	ordered pair $(p,q)$ as well as $^q\Box_p$ (or $_p\Box^q$). 
	This makes sense as there is a one-to-one correspondence between 
	the tuples from $\XX$ and rectangles with endpoints 
in $\XX$ as $\XX$ comes from a permutation sequence. 

	We now show some properties of $\CP(\cdot)$.
	Let $\R_\GG := \{ p \in \GG~|~\lrarrow{\OP(p)}{p} \}$;
	define $\LL_\GG$ similarly.

	\begin{lemma}
	\label{lem:unique-coupling}
	Let $p_1,p_2 \in \R_\GG$. If $p_1 \ne p_2$ then $\CP(p_1) \ne \CP(p_2)$.
	\end{lemma}
	\begin{proof}

	Assume for contradiction that $\CP(p_1) = \CP(p_2)$. Then $p_1$ and $p_2$ are on the 
	same horizontal line. Assume w.l.o.g. that $\lrarrow{p_1}{p_2}$ and $\OP(p_1) = \OP(p_2) = p$. 
	Assume that while processing 
	$p \in \XX$, $\GRE$ marks point $p_1$ and $p_2$ due to 
	unsatisfied rectangle $_p\Box^{q_1}$ and $_p\Box^{q_2}$.
	This implies that $^{q_1}\senwarrow_{q_2}$, as otherwise $_p\Box^{q_2}$ is satisfied by $q_1$ 
	(since $\lrarrow{p_1}{p_2}$, $\lrarrow{q_1}{q_2}$). 
	Also note that $\CP(p_1) = (p, \OP(q_1))$ 
	and $\CP(p_2) = (p,\OP(q_2))$. 
	Since $\Time{q_1} \ne \Time{q_2}$, $\OP(q_1) \ne \OP(q_2)$.
	So $\CP(p_1) \ne \CP(p_2)$, contradicting our assumption. 
	\end{proof}
	This implies that $|\R_\GG| \le |\CP(\GG)|$. 
	By symmetry, we also have $|\LL_\GG| \le |\CP(\GG)|$. 
	This gives 
	\begin{corollary}
	\label{cor:couplingY}
	$|\GG| = |L_\GG| + |R_\GG| \le  2 |\CP(\GG)|$.
	\end{corollary}

	We now show that for any $p \in \GG$, $\CP(p)$ is either in $\MMC$ or in $\MFC$ (thus 
	the third possibility of pairs in Fig.~\ref{fig:mfc-or-mmc}, or its symmetric version, does not arise):

	\begin{lemma}
	\label{lem:coupling-mmc-or-mfc}
	For $p \in \GG$, either $\CP(p) \in \MMC$ or $\CP(p) \in \MFC$. 
	\end{lemma}
	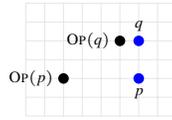
\begin{figure}[H]
\centering
\begin{tikzpicture}

\draw[step=0.25,black,opacity=0.1] (0,0) grid (2,1.5);

\fill [black] (0.5,0.5) circle (2pt) node[left,black,opacity=1] {\tiny{$\OP(p)$}};
\fill [ blue] (1.5,0.5) circle (2pt) node[below,black,opacity=1] {\tiny{$p$}};

\fill [ blue] (1.5,1) circle (2pt) node[above,black,opacity=1] {\tiny{$q$}};
\fill [ black] (1.25,1) circle (2pt) node[left,black,opacity=1] {\tiny{$\OP(q)$}};

\end{tikzpicture}
\caption{Illustration of the bad case which arises in case (3) of the proof of Lemma \ref{lem:coupling-mmc-or-mfc}}
\label{fig:mfc-or-mmc}
\end{figure}
	\begin{proof}
	W.l.o.g., let $\lrarrow{\OP(p)}{p}$ and let $q$ be the first point above $p$. 
	So, $\CP(p) = (\OP(p),\OP(q))$. 
	We have $\CP(p) \notin \MMC$ and $\CP(p) \notin \MFC$ 
	only if $_{\OP(p)}\neswarrow^{\OP(q)}$ and $^{\OP(q)}\nwsearrow_{p}$ (see Fig.~\ref{fig:mfc-or-mmc}).
	So, $\OP(q)$ hides $q$ from $\OP(p)$ (as $\lrarrow{\OP(q)}{q}$ in $_{\OP(p)}\Box^q$) 
	and $p$ is the first point below $q$ such that $^{\OP(q)}\nwsearrow_{p}$ (this 
	follows from the definition of $\CP(p)$ where we said that $q$ is the first point above $p$).
	Then by Lemma~\ref{lem:hidden-right},
	$^{\OP(q)}\nwsearrow_{\OP(p)}$. This contradicts our assumption that  $_{\OP(p)}\neswarrow^{\OP(q)}$.
	So our assumption on the position of $\OP(q)$ must be false.
	\end{proof}

	This shows that $\CP(\GG)$ is partitioned by $\MMC$ and $\MFC$:
	\begin{corollary}
	\label{lem:copulingsamecardinality}
	$|\MMC| + |\MFC| = |\CP(\GG)|$.
	\end{corollary}
\fi

We now show some properties of $\CP(\cdot)$ when the sequence $\XX$ is decomposable.
The following property is independent of whether $\CP(\cdot)$ is in 
$\MMC$ or $\MFC$.

\iflong
	\begin{lemma}
	\label{lem:notoutsideblock}
\else
	\newtheorem*{lem:notoutsideblock}{Lemma \ref{lem:notoutsideblock}}
	\begin{lem:notoutsideblock}	
\fi
Let $p_1 \in \GG$ and $p,q \in \XX$ and $\CP(p_1) = (p,q)$.
If $p \in B$ but $p \ne \Top(B)$, then $q \in B$.
\iflong
	\end{lemma}
\else	
	\end{lem:notoutsideblock}
\fi

\iflong 
	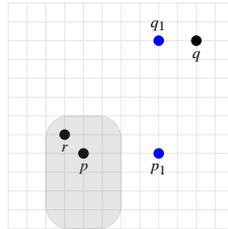
\begin{figure}[H]
\centering
\begin{tikzpicture}
[round/.style={rectangle, rounded corners=3mm, minimum size=20mm, draw=black!50}]

\draw[step=0.25,black,opacity=0.1] (0,0) grid (3,3);

\fill [black] (1,1) circle (2pt) node[below,black,opacity=1] {\tiny{$p$}};
\fill [black] (0.75,1.25) circle (2pt) node[below,black,opacity=1] {\tiny{$r$}};
\fill [ blue] (2,1) circle (2pt) node[below,black,opacity=1] {\tiny{$p_1$}};

\fill [ blue] (2,2.5) circle (2pt) node[above,black,opacity=1] {\tiny{$q_1$}};
\fill [ black] (2.5,2.5) circle (2pt) node[below,black,opacity=1] {\tiny{$q$}};

\draw [round,fill=gray,opacity=0.2] (0.5,0) rectangle (1.5,1.5);






\end{tikzpicture}
\caption{Lemma \ref{lem:notoutsideblock} shows that the scenario depicted in the figure cannot occur, i.e., $p \ne \Top(B)$, $\CP(p_1) = (p,q)$ and $q \notin B$ }
\label{fig:notoutsideblock}
\end{figure}
	\begin{proof}
	Assume for contradiction that $\CP(p_1) = (p,q)$ such that 
	$q \notin B$ (see Fig.~\ref{fig:notoutsideblock} for an illustration).
	W.l.o.g., let $\lrarrow{p}{p_1}$.
	Let $_p\Box^{q_1}$ be the unsatisfied rectangle, that made GREEDY mark $p_1$ while processing $p \in \XX$.
	Note that $\OP(q_1)=q$ can lie to the right or left of $q_1$ or it can be same as $q_1$ (though 
	in Fig.~\ref{fig:notoutsideblock} it lies to the right of $q_1$).

	Let $\Top(B) = r$. Note that $r$ cannot lie above $p$ (as points from $\XX$ come from a 
	permutation sequence) or to the north-east of 
	$p$ as then $_p\Box^{q_1}$ is already arborally satisfied 
	due to $r$. So $^r\nwsearrow_p$ as shown in Fig.~\ref{fig:notoutsideblock}.
	By Lemma \ref{lem:upperbox}, 
	no point lies in $\UPB(B)$. This implies that
	$\GRE$ should find even $_r\Box^{q_1}$ unsatisfied. In that case,
	$\GRE$ should put another marked point, say $r'$ to the right 
	of $r$ below $q_1$. However, this contradicts our deduction 
	that while processing $p$, $\GRE$ found $_p\Box^{q_1}$ unsatisfied (because then $_p\Box^{q_1}$
	is already satisfied by $r'$). So we arrive at a contradiction and
	our assumption that $q \notin B$ must be false.

	\end{proof}
\fi
In words, the above lemma says that for a block $B$ all the points marked 
by $\GRE$ for $p \in B$ have pairs \emph{local} to $B$
if $\Top(B) \ne p$. However, if $\Top(B) = p$, the above 
claim does not hold. In Section \ref{subsec:decomposable-sequence}, 
we saw that there can be many blocks for which $p = \Top(\cdot)$.
Let $\RT(p) = B$. For such $p, B$, we now show that 
all points marked by $\GRE$ will have {\em non-local} pairs.

\iflong
	\begin{lemma} 
	\label{lem:reversetop}
\else
	\newtheorem*{lem:reversetop}{Lemma \ref{lem:reversetop}}
	\begin{lem:reversetop}	
\fi
Let $p,q \in \XX$, $p_1 \in \GG$ and $\RT(p) = B$. If $\CP(p_1) = (p,q)$, then 
$q \in B'$ where $B' \in \SB(B)$.
\iflong
	\end{lemma}
\else	
	\end{lem:reversetop}
\fi
\iflong 
	\begin{proof}
	By definition, if $\CP(p_1) = (p,q)$ then $\Time{q} < \Time{p}$.
	Since $p$ is the first original point in $B$ with the least $y$-coordinate, 
	$q \notin B$. Since $\RT(p)=B$, $p \ne \Top(\PP(B))$. 
	Hence, by Lemma~\ref{lem:notoutsideblock}, $q \in \PP(B)$.
	This implies that $q \in B'$ such that $B' \in \SB(B)$.
	\end{proof}
\fi

\section{Upper bounding $|\MMC|$}
\label{sec:mmc}
In this section, we prove 
\begin{theorem}
\label{thm:mmcfinalbound}
$|\MMC| \leq 2n(k-1)$.
\end{theorem}
Consider a point $p \in \XX$. 
In Section \ref{sec:coupling}, we proved Lemma~\ref{lem:reversetop} which says that 
if $\CP(p_1) = (p,q_1)$ and $\RT(p)=B$, then $q_1 \in B_1$ where $B_1 \in \SB(B)$. We will now prove a certain
strengthening of this lemma.  
Let $\RMMC_p = \{ \CP(p_i) \in \RMMC~|~ \CP(p_i) = (p,q_i)\text{ for some }q_i \in \XX \}$.
In \iflong the next \fi 
section, we prove Lemma~\ref{lem:finalmmc} which essentially says that if $\CP(p_1), \ldots, \CP(p_\ell) \in \RMMC_p$
with $\CP(p_i)=(p, q_i)$, then blocks $B_i$ containing $q_i$ are pairwise distinct. 
By the $k$-decomposability of $\XX$, there are at most $k-1$ siblings of $B$, hence we have $\ell \le k-1$, 
which gives $|\RMMC_p|\le k-1$. 
The same reasoning holds for $\LMMC_p$. And this immediately gives the bound on $|\MMC|$.






The following lemma will be used in proving Lemma~\ref{lem:finalmmc}.

\iflong
	\begin{lemma}
	\label{lem:mmc-witness}
	For $p \in \XX$, let $p_1, p_{2} \in \GG$ be two points 
	such that $p \leftrightarrow p_1 \leftrightarrow p_{2}$ and $\CP(p_1), \CP(p_{2}) \in \RMMC_p$. 
	Let $_p\Box^{r_1}$ and $_p\Box^{r_{2}}$ be the two unsatisfied rectangles that lead
	$\GRE$ to mark $p_1$ and $p_{2}$, respectively, while processing $p$ (thus $^{r_1}\nwsearrow_{r_{2}}$).
	Then there is a point $s \in \XX$ such that $^{r_1}\nwsearrow_{s}$ and 
	($\abarrow{s}{r_{2}}$ or $^{s}\nwsearrow_{r_{2}}$ or $_{r_{2}}\neswarrow^{s}$).   
	\end{lemma} 


	\begin{figure}[H]
\centering
\begin{tikzpicture}
[round/.style={rectangle, rounded corners=3mm, minimum size=20mm, draw=black!50}]

\draw[step=0.25,black,opacity=0.1] (0,0) grid (3,3);

\fill [black] (1,1) circle (2pt) node[below,black,opacity=1] {\tiny{$p$}};

\fill [ blue] (1.75,1) circle (2pt) node[below,black,opacity=1] {\tiny{$p_1$}};
\fill [ blue] (2.25,1) circle (2pt) node[below,black,opacity=1] {\tiny{$p_{2}$}};
\fill [ blue] (2.25,1.75) circle (2pt) node[below,black,opacity=1] {\tiny{$r_{2}$}};
\fill [ black] (0,1.75) circle (2pt) node[below,black,opacity=1] {\tiny{$\OP(r_{2})$}};

\fill [ blue] (1.75,2.5) circle (2pt) node[above,black,opacity=1] {\tiny{$r_1$}};
\fill [ black] (0.25,2.5) circle (2pt) node[above,black,opacity=1] {\tiny{$\OP(r_1)$}};

\draw [round,fill=gray,opacity=0.2] (0.5,0) rectangle (1.5,1.5);
\node at (1.5,-0.5) {(i)};

\begin{scope}[xshift = 100]
\draw[step=0.25,black,opacity=0.1] (0,0) grid (3,3);

\fill [black] (1,1) circle (2pt) node[below,black,opacity=1] {\tiny{$p$}};

\fill [ blue] (1.75,1) circle (2pt) node[below,black,opacity=1] {\tiny{$p_1$}};
\fill [ blue] (2.25,1) circle (2pt) node[below,black,opacity=1] {\tiny{$p_{2}$}};
\fill [ blue] (2.25,1.75) circle (2pt) node[below,black,opacity=1] {\tiny{$r_{2}$}};
\fill [ black] (0,1.75) circle (2pt) node[below,black,opacity=1] {\tiny{$\OP(r_{2})$}};

\fill [ gray] (2.25,3) circle (2pt) node[below,black,opacity=1] {\tiny{$t$}};
\fill [ blue] (1.75,2.5) circle (2pt) node[above,black,opacity=1] {\tiny{$r_1$}};
\fill [ black] (0.25,2.5) circle (2pt) node[above,black,opacity=1] {\tiny{$\OP(r_1)$}};

\draw [round,fill=gray,opacity=0.2] (0.5,0) rectangle (1.5,1.5);
\node at (1.5,-0.5) {(ii)};
\end{scope}

\end{tikzpicture}
\caption{(i) The setting of Lemma \ref{lem:mmc-witness}: (1) $\CP(p_1), \CP(p_{2}) \in \RMMC_p$ and (2) $_p\Box^{r_1}$ and $_p\Box^{r_{2}}$
are two unsatisfied rectangle encountered while processing $p$. (ii) Illustration of the case (4) in the proof when $_{r_1}\swnearrow^{\UP(r_{2})}$}.
\label{fig:witness-mmc}
\end{figure}
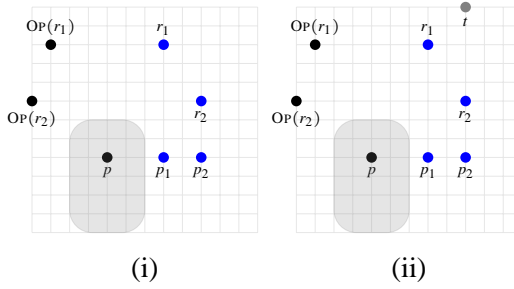
	\begin{proof}
	Note that $r_1, r_{2} \notin \XX$ since $\CP(p_1), \CP(p_{2}) \in \MMC$.
	We consider four cases depending on the position of $\UP(r_{2})$:
	\begin{enumerate}
	\item $\UP(r_{2}) = r_{2}$.

	Since $\CP(p_{2}) \in \MMC$, $r_{2} \notin \XX$. 
	So $\UP(r_{2}) \ne r_{2}$.

	\item $^{r_1}\senwarrow_{\UP(r_{2})}$.

	In this case we can set $s := \UP(r_{2})$. So, we have found an original point $s$ 
	such that ($^{r_1}\senwarrow_{s}$ and $\abarrow{s}{r_{2}}$).

	\item $\lrarrow{r_1}{\UP(r_{2})}$.

	Since $\CP(p_1) \in \MMC$ and $r_1$ is the first point above $p_1$,
	hence $\lrarrow{\OP(r_1)}{r_1}$.
	So $\Time{\UP(r_{2})} \neq \Time{r_1}$; in particular $\lrarrow{r_1}{\UP(r_{2})}$ cannot happen.

	\item $_{r_1}\swnearrow^{\UP(r_{2})}$. 

	Let $t$ be the first point above $r_{2}$ such that $_{r_1}\swnearrow^{t}$
	 (see Fig.~\ref{fig:witness-mmc}(ii)). Since we assumed that $_{r_1}\swnearrow^{\UP(r_{2})}$,  
	$\UP(r_{2})$ is one such candidate. 
	This implies that $r_1$ hides $t$ from $p$(as $r_1 \in~ (_p\Box^t)^{\circ}$). Let $t'$ be the first point below $t$ such that
	$^{r_1}\senwarrow_{t'}$. Such a point exists as $r_{2}$ is one such candidate. By Lemma \ref{lem:hidden-right},
	$^{r_1}\senwarrow_{\OP(t')}$.
	In that case,  $t' \ne r_{2}$ since $^{\OP(r_{2})}\nwsearrow_p$ (since $\CP(p_{2}) \in \RMMC$).
	So $\abarrow{t}{t'}$ and $\abarrow{t'}{r_{2}}$. 
	This implies that either $^{\OP(t')}\nwsearrow_{r_{2}}$ or $_{r_{2}}\neswarrow^{\OP(t')}$.
	So we have found an original point $s := \OP(t')$ such that 
	$^{r_1}\nwsearrow_{s}$ and ($^{s}\nwsearrow_{r_{2}}$ or $_{r_{2}}\neswarrow^{s}$).
	\end{enumerate}

	\end{proof}
\fi

\iflong
	The following lemma will allow us to use the $k$-decomposability of $\XX$ to get at upper bound on 
	$|\MMC|$. 
\fi


\begin{lemma}
\label{lem:finalmmc}
Let $p \in \XX$ and $\RT(p) = B$.
Let $p_1 \leftrightarrow p_2 \leftrightarrow \dots \leftrightarrow p_\ell$ be
such that $\CP(p_i) \in \RMMC_p$ for all $1 \le i \le \ell$. 
Assume that $\GRE$ marks points 
$p_1 \leftrightarrow p_2 \leftrightarrow \dots \leftrightarrow p_\ell$
due to unsatisfied rectangles $_p\Box^{r_1},_p\Box^{r_2}, \dots, _p\Box^{r_\ell}$, respectively.
By Lemma~\ref{lem:reversetop}, $\OP(r_i) \in B_i$ for $B_i \in \SB(B)$. 
Then $B_i \ne B_j$ for $i \neq j$. 
\end{lemma}

\iflong
	\begin{proof}

	Consider $\lrarrow{p_i}{p_j}$. One can check that $\Time{r_i} < \Time{r_j}$.
	By Lemma \ref{lem:mmc-witness},
	there exists an original point $s$ such that  $^{r_i}\nwsearrow_{s}$ and 
	($\abarrow{s}{r_{j}}$ or $^{s}\nwsearrow_{r_{j}}$ or $_{r_{j}}\neswarrow^{s}$). 

	Since $^{\OP(r_{i})}\nwsearrow_p$ and $^{\OP(r_{j})}\nwsearrow_p$ (as $\CP(p_i), \CP(p_{j}) \in \MMC$), 
	this implies that  
	 $^{\OP(r_i)}\nwsearrow_{s}$ and  $_{\OP(r_{j})}\neswarrow^{s}$. 
	Also, 
	we have $p \in B$ with $\Key{\OP(r_i)} < \Key{p} < \Key{s}$ and $p \notin B_i$
	(since $B_i$ is a siblings of $B$).
	By the definition of a block, all the points in $B_i$ 
	should be contiguous on the $x$-axis. 
	So, $s \notin B_i$. 

	Thus we have $s \notin B_i$ and $\Time{\OP(r_i)} < \Time{s} < \Time{\OP(r_{j})}$.
	Again, by the definition of a block, all the points in $B_i$  
	should have contiguous time interval. 
	So, $\OP(r_{j})$ 
	does not lie in $B_i$ and $B_i \ne B_{j}$.






	\end{proof}
\fi

Since there are at most $k-1$ sibling of $B$, by Lemma~\ref{lem:finalmmc}, $|\RMMC_p| \le k-1$.
So, $|\RMMC| = \sum_{p \in \XX} |\RMMC_p| \le n(k-1)$.
Similarly, we can show that $|\LMMC| \leq n(k-1)$. Since $|\MMC| = |\RMMC| + |\LMMC|$, 
we have proved Theorem~\ref{thm:mmcfinalbound}.
\section{Improving the bound on $|\mmc|$}
\label{sec:mmc-betterbound}
In this section we prove:
\begin{theorem}
\label{thm:boundmmc}
$|\MMC| \le 14 n \log k$.
\end{theorem} 
To improve the the bound on $|\MMC|$, we need to show some more properties of a 
block in $\TD(\XX)$.
\subsection{Properties of a Block}
Let $\MAXT(B) := argmax_{z \in B} \{ \Time{z} \}$. That is, $\MAXT(B)$ is an original point
in $B$ that has the maximum y-coordinate. We can similarly define $\MINT(B)$. Note that 
$\Top(B) = \MINT(B)$. Let $\MAXK(B) := argmax_{z \in B} \{ \Key{z} \}$ and 
$\MINK(B) := argmin_{z \in B} \{ \Key{z} \}$.

\begin{center}
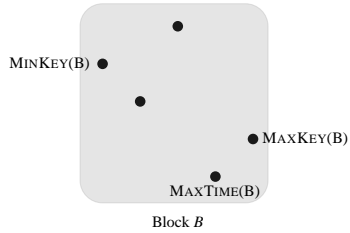
\begin{figure}[H]
\centering
\begin{tikzpicture}
[round/.style={rectangle, rounded corners=3mm, minimum size=20mm, draw=black!50}]

\fill [black] (1,2) circle (2pt);
\fill [ black] (0,1.5) circle (2pt);
\fill [ black] (0.5,1) circle (2pt);
\fill [ black] (2,0.5) circle (2pt);
\fill [ black] (1.5,0) circle (2pt);
\draw [round,fill=gray,opacity=0.2] (-0.30,2.3) rectangle (2.2,-.35);
\node at (-.7,1.5) {\tiny{\MINK(B)}};
\node at  (2.7,0.5)  {\tiny{\MAXK(B)}};
\node at  (1.5,-.2)  {\tiny{\MAXT(B)}};

\node at  (1,-.6)  {\tiny{Block $B$}};

\end{tikzpicture}
\caption{Pictorial representation of $\MAXT(B), \MAXK(B)$ and $\MINK(B)$}
\label{fig:blockdef}
\end{figure}
\end{center}

\begin{definition}(Left and Right points of block $B$)
Let $\RT(p) = B$. Define $\Left(B)$ to be the first point to the left of $p$ marked 
by $\GRE$ while processing $p$. Define $\Right(B)$ in a symmetric fashion.
\end{definition}
Note that $\Left(B)$, if it exists, satisfies $\Key{\Left(B)} < \Key{\MINK(B)}$ since by Observation \ref{obs:nothing-above},
$\GRE$ cannot put any point above any other original point in $B$. Similarly
$\Key{\Right(B)} > \Key{\MAXK(B)}$. This implies that $\GRE$ does not put any point 
in $\BOX(B)$ while processing $\Top(B)$.

\begin{observation} 
\label{obs:topnotinbox}
$\GRE$ does not put any marked points in $\BOX(B)$ while processing $\Top(B)$.
\end{observation}

\begin{lemma}
\label{lem:blockleftright} 
Let $\RT(r) = B'$ and $\PP(B') = B$.
So $r \ne \Top(B)$.
While processing $r$, $\GRE$ can put marked point (1)Below $\Left(B)$ or/and (2)Below $\Right(B)$ or/and (3) In $\BOX(B)$.  
\end{lemma}
\begin{center}
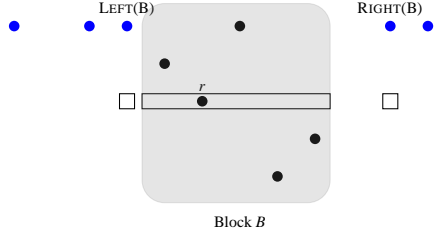
\begin{figure}[H]
\centering

\begin{tikzpicture}
[round/.style={rectangle, rounded corners=3mm, minimum size=20mm, draw=black!50}]

\fill [black] (1,2) circle (2pt);

\fill [blue] (-0.5,2) circle (2pt) node[above,black]{\tiny{\Left(B)}};
\fill [blue] (-1,2) circle (2pt);

\fill [blue] (3,2) circle (2pt) node[above,black]{\tiny{\Right(B)}};
\fill [blue] (-2,2) circle (2pt);
\fill [blue] (3.5,2) circle (2pt);

\fill [ black] (0,1.5) circle (2pt);
\fill [ black] (0.5,1) circle (2pt) node[above,black]{\tiny{$r$}};
\draw  (-0.6,0.9) rectangle (-.4,1.1);
\draw  (-0.3,0.9) rectangle (2.2,1.1);
\draw  (2.9,0.9) rectangle (3.1,1.1);

\fill [ black] (2,0.5) circle (2pt);
\fill [ black] (1.5,0) circle (2pt);
\draw [round,fill=gray,opacity=0.2] (-0.30,2.3) rectangle (2.2,-.35);

\node at  (1,-.6)  {\tiny{Block $B$}};

\end{tikzpicture}
\caption{$\Left(B)$ and $\Right(B)$ of a block $B$. Note that $\GRE$ cannot put any point in $\BOX(B)$ while
processing $\Top(B)$. Lemma \ref{lem:blockleftright} shows that while processing any point $r \ne \Top(B)$, $\GRE$
can put points only in the rectangles shown in the figure.}
\label{fig:blockleftright}
\end{figure}
\end{center}

\begin{proof}
Assume that $\GRE$ puts a point $r_1$ while processing $r$.
Consider the following two cases for contradiction:
\begin{enumerate}
\item $\Key{r_1} < \Key{\Left(B)}$

$\UP(r_1)$ cannot satisfy the property $_{\UP(r_1)}\swnearrow^{\Left(B)}$ as then $\Key{\UP(r_1)} < \Key{\MINK(B)}$.
By definition of decomposability, all the original points with time between $[ \Time{\Top(B)}, \Time{\MAXT(B)} ]$
should have their key value $\ge \Key{\MINK(B)}$. However, $\Key{\UP(r_1)} < \Key{\Left(B)} <  \Key{\MINK(B)}$.

So assume that $^{\UP(r_1)}\senwarrow_{\Left(B)}$. Let $r'$ be the first point above $r_1$ with the property 
that ($^{r'}\senwarrow_{\Left(B)}$ or $\lrarrow{r'}{\Left(B)}$). Such a point exists as $\UP(r_1)$ is one such candidate. 
This implies that $\Left(B)$ hides $r'$ from $r$ as $\Left(B)$ lies in $^{r'}\Box_r$.
Let $r''$ be the first point below $r'$. Such a point exists as $r_1$ itself is one such candidate.
Also, note that due to the way $r'$ is defined, $_{r''}\swnearrow^{\Left(B)}$. 
By Lemma \ref{lem:hidden}, $_{\UP(r'')}\swnearrow^{\Left(B)}$. This again leads to contradiction as $\Time{\UP(r'')}$ is between 
$[ \Time{\Top(B)}, \Time{\MAXT(B)} ]$ and its key value $\Key{\UP(r'')}< \Key{\Left(B)} <  \Key{\MINK(B)}$.

\item $\Key{\Left(B)} < \Key{r_1} < \Key{\MINK(B)}$

Let $r'$ be a point above $r_1$ with least $y$ co-ordinate and the property that  $^{\Left(B)}\nwsearrow_{r'}$. 
If no such point exists then define $r' := r_1$. Let $\GRE$ mark point $r'$ while processing $\OP(r') \in B$ due to an unsatisfied 
rectangle $^{r''}\Box_{\OP(r')}$. By definition of $\Left(B)$, $r''$ cannot lie to the right of $\Left(B)$. 
So, $_{\Left(B)}\neswarrow^{r''}$. This implies that there exists a non-top element of $B$, $\OP(r')$
with $\RR(r') = (\OP(r'),\OP(r''))$ and $\OP(r'') \notin B$. This contradict Lemma \ref{lem:notoutsideblock}
\end{enumerate}
The case when $\Key{r_1} > \Key{\Right(B)}$ and $\Key{\MAXK(B)} >\Key{r_1} > \Key{\Right(B)} $
are symmetric to above two cases respectively.
\end{proof}

We now need to calculate the number of points added by $\GRE$ in $\BOX(B)$ while processing $r$.
To this end, we define some notations.
\begin{definition}
Let $\RG(B) := \{(a,b) | \UP(a) \in  B$ \text{and} $b> \Time{\MAXT(B)}$ \text{and} $b < \Time{\MAXT(\PP(B))} \}$
\end{definition}
In words, $\RG(B)$ denote the region below $\BOX(B)$ till the last original point in $\PP(B)$.
If $\BB$ is the set of blocks, then define $\RG(\BB) := \cup_{B \in \BB} \RG(B)$.

\begin{lemma}
Let $\RT(r) = B'$ and $\PP(B') = B$. Then while processing $r$, $\GRE$ can put
marked point only (1)below $\Left(B)$ and/or (2)Below $\Right(B)$ and/or (3) in $\RG(B'')$ where $B'' \in \SB(B')$.
\end{lemma}
\begin{proof}
By Lemma \ref{lem:blockleftright}, while processing $r$,
$\GRE$ can put marked point in $\BOX(B)$, below $\Left(B)$ and below $\Right(B))$. 
However by observation \ref{obs:topnotinbox}, $\GRE$ cannot 
put any marked point in $\BOX(B')$ while processing $r$. 

Let $B'' \in \SB(B')$. 
By Lemma \ref{lem:upperbox}, $\GRE$ cannot put any marked point in
$\UPB(B'')$ while processing $r$. Also, $\GRE$ cannot put any marked point in $\BOX(B'')$ while processing $r$. 
This implies that while processing $r$,
all the points marked by $\GRE$ in $\BOX(B)$ are in $\RG(B'')$ where $B'' \in \SB(B')$.
\end{proof}

\begin{center}
\begin{figure}[H]
\centering
\begin{tikzpicture}[scale=0.4,round/.style={rectangle, rounded corners=3mm, minimum size=20mm, draw=black!50}]

\draw [round,fill=gray,opacity=0.1] (-6,8) rectangle (8,-6);

\fill [black] (1,2) circle (2pt);
\fill [ black] (0,1.5) circle (2pt);
\fill [ black] (0.5,1) circle (2pt);
\fill [ black] (2,0.5) circle (2pt);
\fill [ black] (1.5,0) circle (2pt);
\draw [round,fill=gray,opacity=0.2] (-0.30,2.3) rectangle (2.2,-.35);

\node at  (1,2.6)  {\tiny{Block $B$}};
\node at  (1,8.5)  {\tiny{Block $\PP(B)$}};
\node at  (1,-6.5)  {\tiny{$\RG(B)$}};

\begin{scope}[shift={(2.5,2.5)}]
\fill [black] (1,2) circle (2pt);
\fill [ black] (0,1.5) circle (2pt);
\fill [ black] (0.5,1) circle (2pt);
\fill [ black] (2,0.5) circle (2pt);
\fill [ black] (1.5,0) circle (2pt);
\draw [round,fill=gray,opacity=0.2] (-0.30,2.3) rectangle (2.2,-.35);
\end{scope}

\begin{scope}[shift={(-2.5,-2.5)}]
\fill [black] (1,2) circle (2pt);
\fill [ black] (0,1.5) circle (2pt);
\fill [ black] (0.5,1) circle (2pt);
\fill [ black] (2,0.5) circle (2pt);
\fill [ black] (1.5,0) circle (2pt);
\draw [round,fill=gray,opacity=0.2] (-0.30,2.3) rectangle (2.2,-.35);
\end{scope}

\begin{scope}[shift={(-5,5)}]
\fill [black] (1,2) circle (2pt);
\fill [ black] (0,1.5) circle (2pt);
\fill [ black] (0.5,1) circle (2pt);
\fill [ black] (2,0.5) circle (2pt);
\fill [ black] (1.5,0) circle (2pt);
\draw [round,fill=gray,opacity=0.2] (-0.30,2.3) rectangle (2.2,-.35);
\end{scope}

\begin{scope}[shift={(5,-5)}]
\fill [black] (1,2) circle (2pt);
\fill [ black] (0,1.5) circle (2pt);
\fill [ black] (0.5,1) circle (2pt);
\fill [ black] (2,0.5) circle (2pt);
\fill [ black] (1.5,0) circle (2pt);
\draw [round,fill=gray,opacity=0.2] (-0.30,2.3) rectangle (2.2,-.35);
\end{scope}

\draw [round,fill=gray,opacity=0.6] (-0.30,-6) rectangle (2.2,-.35);
\end{tikzpicture}
\caption{$B$ has four siblings. $\RG(B)$ is the region shaded below $\BOX(B)$}
\label{fig:blockregion}
\end{figure}
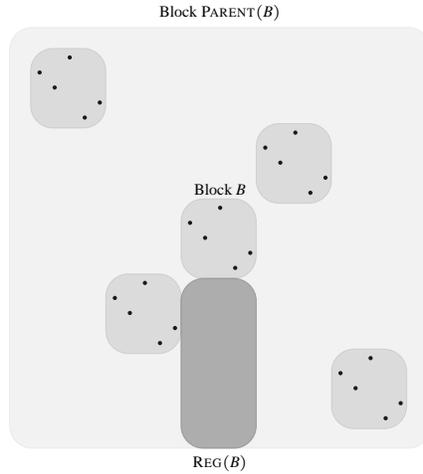
\end{center}
We will now describe properties of the points that are added in $\RG(B)$.
\begin{definition}
A point $p$ is a key-new point in $\RG(B)$ 
if there is no point $q \in \RG(B)$ such that (1)$\Key{p} = \Key{q}$
and $\Time{q} < \Time{p}$, else it is called key-old. 
\end{definition}

In other words, if no point exists in $\RG(B)$ with the same key as $\Key{p}$ before
time $\Time{p}$, then $p$ is key-new. 

\begin{definition}
Let $p \in \XX$ be an original point with least $\Time{p}$ and the property 
(1) $\Time{p} > \Time{\MAXT(B)}$ and (2)$\Key{p} < \Left(B)$ or $\Left(B) < \Key{p} < \MINK(B) $. 
Then we say that $p$ is the left-relative of $B$ or $p \in \LeftRel(B)$.
Similarly define $\RightRel(B)$ in a symmetric fashion.
\end{definition}

In words, $\LeftRel(B)$ contains the first original point in the 
sequence $\XX$ that comes after all the points in $B$ and has key 
strictly than  $\Key{\Left(B)}$ or in the range $[\Key{\Left(B)}, \Key{\MINK(B)}]$.
Note that there are at most 2 original points in the set $\LeftRel(B)$.
\begin{lemma}
\label{lem:pointotherthanrel} 
Let $q$ be an original point that satisfies $\Key{q} < \Left(B)$
and $\Time{q} > \Time{p}$ where $p \in \LeftRel(B)$ and $\Key{p} < \Key{\Left(B)}$. Then 
$\GRE$ cannot mark any key-new point in $\RG(B)$ while processing $q$.
\end{lemma}

\begin{center}
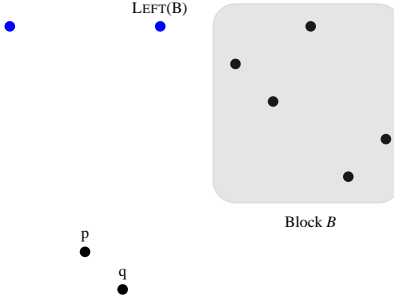
\begin{figure}[H]
\centering
\begin{tikzpicture}
[round/.style={rectangle, rounded corners=3mm, minimum size=20mm, draw=black!50}]

\fill [black] (1,2) circle (2pt);

\fill [blue] (-1,2) circle (2pt) node[above,black]{\tiny{\Left(B)}};
\fill [blue] (-3,2) circle (2pt);

\fill [ black] (0,1.5) circle (2pt);
\fill [ black] (0.5,1) circle (2pt);
\fill [ black] (2,0.5) circle (2pt);
\fill [ black] (1.5,0) circle (2pt);
\draw [round,fill=gray,opacity=0.2] (-0.30,2.3) rectangle (2.2,-.35);

\fill [black] (-2,-1) circle (2pt) node[above,black]{\tiny{p}};
\fill [black] (-1.5,-1.5) circle (2pt) node[above,black]{\tiny{q}};
\node at  (1,-.6)  {\tiny{Block $B$}};

\end{tikzpicture}
\caption{$p \in \LeftRel(B)$. Lemma \ref{lem:pointotherthanrel} states that while processing 
any other point $q$ that arrives after $p$ and is to the left of $\Left(B)$, $\GRE$ cannot 
mark any key-new point in $\RG(B)$.}
\label{fig:pointotherthanrel}
\end{figure}
\end{center}

\begin{proof}
Assume for contradiction that $\GRE$ put a key-new point $s$ while processing $q$
in $\RG(B)$ due to unsatisfied rectangle $_q\Box^{s'}$. This implies that 
$\RR(s) = (q, \OP(s'))$. Since we have assumed that $s$ is a key-new point 
$s'$ cannot lie in $\RG(B)$, so it lies above $\RG(B)$. By Lemma \ref{lem:upperbox}, no point can lie above $\BOX(B)$.
So $s'$ lies in $\BOX(B)$. We first note that $\OP(s') \ne \Top(B)$ as then $_q\Box^{s'}$ is satisfied by 
$\Left(B)$. 

Let us assume that $\OP(s') \ne \Top(B)$.
Since $_q\Box^{s'}$ was an unsatisfied rectangle, there exists no 
point in it when $\GRE$ processes $q$. Consider the point $p \in \LeftRel(B)$ with $\Key{p} < \Key{\Left(B)}$.
By definition, $\Time{p} < \Time{q}$. 
Consider the rectangle $_p\Box^{s'}$. We claim that $_p\Box^{s'}$ is an unsatisfied rectangle as (1) there are no point 
in $_q\Box^{s'}$ when $\GRE$ processes $p$ and (2) By Lemma \ref{lem:blockleftright}, 
for any point $r \ne \Top(B)$, $\GRE$ does not put any 
point to the left of $\Left(B)$ and (3) There is no point $r$ with the property that $\Time{r} > \Time{\MAXT(B)}$ and 
$_p\neswarrow^r$ and $^r\senwarrow_q$. Indeed, if such a point exists, then define $r$ be a point which satisfy the above property and has 
the least $y$-coordinate. If $^{\OP(r)}\senwarrow_q$, then $\OP(r) \in \LeftRel(b)$ and not $p$ which contradicts our assumption.
$\OP(r)$ cannot lie in  $_q\Box^{s'}$ and in $\RG(B)$. So $^{\MAXT(B)}\senwarrow_{\OP(r)}$ and $\GRE$ markes the point
$r$ while processing $\OP(r)$ due to unsatisfied rectangle, say $^{r'}\Box_{\OP(r)}$. Note that $_{r'}\neswarrow^{\MAXT(B)}$, as otherwise
$^{r'}\Box_{\OP(r)}$ is satisfied by $\MAXT(B)$. However, this implies that $r'$ is the point with least $y$-coordinate and the property 
that $\Time{r'} > \Time{\MAXT(B)}$ 
$_p\neswarrow^{r'}$ and $^{r'}\senwarrow_q$. This contradicts the definition of $r$.

This implies that $_p\Box^{s'}$ is unsatisfied when $\GRE$ processes $p$. So $\GRE$ should put a point, say $s''$ below $s'$. However, this 
implies that $s''$ is a key-new point and not $s$ contradicting the assumption of the lemma.
\end{proof}

Similar to the above lemma one can also prove the following lemma.
\begin{lemma}
\label{lem:pointotherthanrelsymmetric}
Let $q$ be an original point that satisfies $\Left(B) < \Key{q} < \MINK(B)$
and $\Time{q} > \Time{p}$ where $p \in \LeftRel(B)$ and $\Left(B) < \Key{p} < \MINK(B)$. Then $\GRE$ cannot mark any key-new point in $\RG(B)$ while 
processing $q$.
\end{lemma}

Similarly one can prove the symmetric versions of the above two lemmas.
We now describe the implications of the above two lemmas. The lemma suggest that at most 
four points ( 2 in $\LeftRel(B)$ and 2 in $\RightRel(B)$) can be responsible for adding 
key-new points in $\RG(B)$.

\begin{corollary}
\label{cor:pointotherthanrelimplication}
$\GRE$ can mark key-new points in $\RG(B)$ while processing points in $\LeftRel(B) \cup \RightRel(B)$ only.
\end{corollary}

We define some more properties of the points added to $\RG(B)$.
\begin{definition}
A key $b$ is live in $\RG(B)$ at time $t$ if there is a point $p \in \RG(B)$ such that $\Key{p}=b$ and $\Time{p} < t$ and 
there exists no point $p' \in \RG(B)$ such that
(1) $\Key{p'} < \Key{p}$ or $\Key{p'} > \Key{p}$ and (2) $t > \Time{p'} \ge \Time{p}$. We say
that key $b$ is not-live in $\RG(B)$ if such a point $p'$ exists. 
A point $p$ is said to be key-live at time $t$ in $\RG(B)$ if key $\Key{p}$ is live in $\RG(B)$ at time $t$, other-wise it is key-not-live.
\end{definition}

\begin{center}
\begin{figure}[H]
\centering
\begin{tikzpicture}[round/.style={rectangle, rounded corners=3mm, minimum size=20mm, draw=black!50}]
\fill [black] (1,2) circle (2pt);
\fill [ black] (0,1.5) circle (2pt);
\fill [ black] (0.5,1) circle (2pt);
\fill [ black] (2,0.5) circle (2pt);
\fill [ black] (1.5,0) circle (2pt);
\node at (1,2.5) {\tiny{Block $B$}};
\draw [round,fill=gray,opacity=0.2] (-0.30,2.3) rectangle (2.2,-.35);

\draw [round,fill=gray,opacity=0.6] (-0.30,-4) rectangle (2.2,-.35);
\node at (1,-4.2) {\tiny{$\RG(B)$}};

\fill [ blue ] (0,-1.5) circle (2pt) node[black,above] {\tiny{$p_1$}};
\fill [ blue ] (0.5,-2) circle (2pt) node[black,above] {\tiny{$p_2$}};
\fill [ blue ] (1,-1) circle (2pt) node[black,above] {\tiny{$p_3$}};
\fill [ blue ] (1.5,-1.5) circle (2pt) node[black,above] {\tiny{$p_4$}};
\fill [ blue ] (2,-2) circle (2pt) node[black,above] {\tiny{$p_5$}};

\draw (-2,-2.5) -- (4,-2.5) node[above]{\tiny{Time $t$}};
\end{tikzpicture}
\caption{Key $\Key{p_1}, \Key{p_2}$ and $\Key{p_5}$ are live at time $t$ and Key $\Key{p_3}$ and $\Key{p_4}$ 
are not live. Similarly point $p_1,p_2$ and $p_5$ are key-live and point $p_3$ and $p_4$ are 
key-not-live.}
\label{fig:blockleftright}
\end{figure}
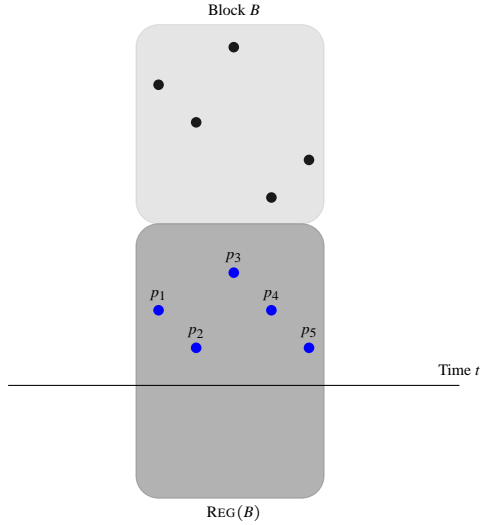
\end{center}


By Lemma \ref{lem:pointotherthanrel}, only points in $\LeftRel(B) \cup \RightRel(B)$ can add key-new points to $\RG(B)$.
Let $p \in \LeftRel(B)$ such that $\GRE$ adds a key-new point $q$ in $\RG(B)$ while processing $p$
due to unsatisfied rectangle $_p\Box^{q'}$. Since $q$ is key-new in $\RG(B)$, $q'$ must lie in $\BOX(B)$.
This implies that $\OP(q') \in B$. Since $p$ is a left-relative of $B$, it lies to the left of $\MINK(B)$.
This implies that $\RR(q) \in \MFC$. So, we have proved the following lemma:

\begin{lemma}
\label{lem:key-newinmfc} 
All key-new points added by original points in $\LeftRel(B)$ in $\RG(B)$ have $\RR(\cdot) \in \MFC$.
\end{lemma}

One can also prove the symmetric version of the above lemma in a similar way.
\begin{lemma}
\label{lem:key-newinmfcsymmetric}
All key-new points added by original points in $\RightRel(B)$ in $\RG(B)$ have $\RR(\cdot) \in \MFC$.
\end{lemma}

\subsection{Improving the bound on $|\MMC|$ to $O(n \log k)$}
Let $B$ be any block of $\TD(\XX)$ that has at most $\ell (\ell < k)$ children.
Let these children be $B_1, B_2, \dots, B_{\ell}$ where $\Key{\MAXK(B_i)} < \Key{\MINK(B_{i+1})}$, or 
in words, each key of the block $B_i$ is smaller than each key of block $B_{i+1}$.
We recursively partition $\ell$ blocks into two half till a singleton block is obtained. This partition can be 
represented as a tree called the $\PT(B)$. At the root of $\PT(B)$ is the 
region $\RG({B_1,B_2,\dots,B_{\ell}})$. We then divide this region into two half\footnote{we drop the floor notation for readability} :
$\RG(B_1,B_2,\dots, B_{\ell/2})$ and $\RG(B_{\ell/2+1}, B_{\ell/2+2}, \dots, B_{\ell})$.
In Lemma \ref{lem:partitionlemma}, we will show that the total number of points (with $\RR(\cdot) \in \MMC$) added by $\GRE$ 
in $\RG(B_{\ell/2+1}, B_{\ell/2+2}, \dots, B_{\ell})$ while processing points in\\ 
$\{\Top(B_1), \Top(B_2),\dots,\Top(B_{\ell/2})\} \setminus \Top(B)$ is $O(\ell)$. 
Similarly, the total number of points (with $\RR(\cdot) \in \MMC$) added by $\GRE$ in $\RG({B_1,B_2,\dots,B_{\ell}})$ while processing 
original points in \\ $\{ \Top(B_{\ell/2+1}), \Top(B_{\ell/2+2}),\dots,\Top(B_{\ell}) \} \setminus \Top(B)$ is $O(\ell)$.
Note that the top element of the block $B$ is exempted as it cannot add any point in the 
$\RG({B_1,B_2,\dots,B_{\ell}})$.
In general we have to show the following lemma:

\begin{lemma}
\label{lem:partitionlemma}
Let $B_i,B_{i+1},\dots, B_{i+2m-1}$ be the set of consecutive children blocks of $B$ in $\TD(\XX)$.
The number of points added by $\{ \Top(B_{i}), \Top(B_{i+1}), \dots, \Top(B_{i+m-1}) \} \setminus \Top(B)$ 
in $\RG(B_{i+m}, B_{i+m+1}, \dots, B_{i+2m-1})$ with $\RR(\cdot) \in \MMC$ is $\le 12m$.
\end{lemma}

By symmetry one can also show that the number of points added by \\ $\{ \Top(B_{i+m}), \Top(B_{i+m+1}), \dots, \Top(B_{i+2m-1}) \} \setminus \Top(B)$ 
in $\RG(B_{i}, B_{i+1}, \dots, B_{i+m-1})$ with $\RR(\cdot) \in \MMC$ is $\le 12m$.

Let $Y(B) := T(B_1,B_2,\dots,B_{\ell})$ be the total number of points (with $\RR(\cdot) \in \MMC$) added by 
$\GRE$ while processing points in $\{ \Top(B_1), \Top(B_2),\dots, \Top(B_{\ell}) \} \setminus \Top(B)$ in $\RG(B_1,B_2,\dots, B_{\ell})$. 
It can be calculated as follows:
$Y(B) = T(B_1,B_2,\dots,B_{\ell})$ = $T(B_1,B_2,\dots,B_{\ell/2})$ + $T(B_{\ell/2+1},B_{\ell/2+2},\dots,B_{\ell})$ + $12 \ell$.
This would imply that $Y(B) = T(B_1,B_2,\dots,B_{\ell}) \le 12 \ell \log \ell$.

We would charge these $12 \ell \log \ell$ points 
to the following $\ell-1$ original points: \\ $\{ \Top(B_1), \Top(B_2), \dots, \Top(B_{\ell}) \}\setminus {\Top(B)}$.
That is, each top point of children of block $B$ except one gets $\le 13 \log \ell$ charge.

We are now ready to prove Theorem~\ref{thm:boundmmc}
\begin{proof}
Let $\RT(r) = B'$ and $\PP(B') = B$. By Lemma \ref{lem:blockleftright}, $\GRE$ can put at most two point below $\Left(B)$
and $\Right(B)$. And by the analysis above, the amortized number of points (with $\RR(\cdot) \in \MMC$) added by $\GRE$ in 
$\BOX(B)$ while processing $r$ with $\RR(\cdot) \in \MMC$ is $13 \log \ell \le 13 \log k$ where $\ell$ is the number of children of $B$. So, 
amortized number of points added by $r = 2 +  13 \log k$. So $|\MMC| \le 14 n \log k$.
\end{proof}

The rest of the section is devoted in proving Lemma \ref{lem:partitionlemma}.

\subsection{Proof of Lemma \ref{lem:partitionlemma}}
\label{sec:proofpartitionlemma}
Let $B_1, B_2, \dots, B_l$ be the children of $B$
in $\TD(\XX)$. Let $B_i, B_{i+1},\dots,B_{i+2m-1}$ be the consecutive 
$2m$ children of $B$. 
Let $\BB_{\ell} = \{B_i,B_{i+1},\dots,B_{i+m-1}\}$ and $\BB_r = \{B_{i+m}, B_{i+m+1}, \dots, B_{i+2m-1}\}$.
Let $\BB_{ll} = \{B_1,B_2,\dots,B_{i-1}\}$ and 
$\BB_{rr} = \{ B_{i+2m}, B_{i+2m+1}, \dots, B_{\ell} \}$.
Let $C_t$ denote the set of key-live points in $\RG(\BB_r)$ after 
$\GRE$ finished processing all the points till step $t$.

\begin{lemma}
\label{lem:key-liveleftblocks} 
Let $B_j \in \{ \BB_{\ell}, \BB_{\ell\ell} \}$. 
Let $N_{\Top(B_j)}, R_{\Top(B_j)}$ be the set of key-new and key-old points added by 
$\GRE$ in $\RG(\BB_r)$ while processing $\Top(B_j)$. Then $C_{\Time{\MAXT(B_j)}+1} \le C_{\Time{\MINT(B_j)}} - \max\{0,(|R_{\Top(B)}|-2)\} + 2|\{B_k \in \BB_r | {\Top(B)} \in \LeftRel(B_k) \}| $. 
\end{lemma}

\begin{proof}
Let $\lrarrow{\lrarrow{\lrarrow{p_1}{p_2}}{\dots}}{p_n}$ be the set of points added by $\GRE$ in $\RG(\BB_r)$
while processing $\Top(B_j)$. Note that the key $\{\Key{p_1},  \Key{p_2}, \dots, \Key{p_n}\}$ 
are live in $\RG(\BB_r)$ at time $\Time{\Top(B_j)}$.
By definition, only $p_1$ and $p_n$ can be live at time $\Time{\Top(B_j)}+1$ as 
each other point $p_i$ is hidden by the point $p_{i-1}$ and $p_{i+1}$.
$p_1$ and $p_n$ can be key-new or key-old point. Consider only the points in 
$R_{\Top(B_j)}$. All the keys in $R_{\Top(B_j)}$ are live at time $\Time{\Top(B_j)}$ and only 2 of these
($p_1$ and $p_n$) can be live at time $\Time{\Top(B_j)}+1$. So, the number of key-live points 
decrease at least by $\max\{0,|R_{\Top(B_j)}|-2\}$ after processing $\Top(B_j)$.

Consider the points in $N_{\Top(B_j)}$.
By Lemma \ref{lem:pointotherthanrel}, if $\Top(B_j) \in \LeftRel(B_k)$ for $B_k \in \BB_r$, then $\GRE$ can add 
key-new points out of which only only at most two (possibly $p_1$ and $p_n$) are live.
So the total number of key-live points added by $\GRE$ while processing $p$
is at most $2|\{B_k \in \BB_r | \Top(B_j) \in \LeftRel(B_k) \}|$. This implies 
$C_{\Time{\Top(B_j)}} \le C_{\Time{\Top(B_j)}-1} - \max\{0,(|R_{\Top(B_j)}|-2)\} + 2|\{B_k \in \BB_r | {\Top(B_j)} \in \LeftRel(B_k) \}|$.

Consider a point $q \in B_j$ such that $q \ne \Top(B_j)$.
By Lemma \ref{lem:blockleftright}, all the points added by $\GRE$ 
while processing $q$ are either in $\BOX(B_j)$ or below $\Left(B_j)$ and $\Right(B_j)$.
Note that only $\Right(B_j)$ can lie in $\RG(\BB_r)$. Whenever $\GRE$ adds a point below $\Right(B_j)$, 
key $\Key{\Right(B_j)}$ is live at time $\Time{q}$. So $\GRE$ does not increase or decrease key-live points while processing $q$.
This implies  that  $C_{\Time{\MAXT(B_j)}+1} \le C_{\Time{\Top(B_j)}} - \max\{0,(|R_{\Top(B_j)}|-2)\} + 2|\{B_j \in \BB_r | {\Top(B_j)} \in \LeftRel(B_j) \}| $
\end{proof}

Similar, we can prove the following lemma.
\begin{lemma}
\label{lem:key-liverightblock} 
Let $B_j \in \{ \BB_{rr} \}$. 
Let $N_{\Top(B_j)}, R_{\Top(B_j)}$ be the set of key-new and key-old points added by 
$\GRE$ in $\RG(\BB_r)$ while processing $\Top(B_j)$. Then $C_{\Time{\MAXT(B_j)}+1} \le C_{\Time{\Top(B_j)}} - \max\{0,(|R_{\Top(B_j)}|-2)\} + 2|\{B_j \in \BB_r | \Top(B_j) \in \RightRel(B_j) \}| $. 
\end{lemma}

\begin{lemma}
\label{lem:key-livesameblock} 
Let $B_j \in \BB_r$. Then $C_{\Time{\MAXT(B_j)}+1} \le C_{\Time{\MINT(B_j)}} + 4$.
\end{lemma}

\begin{proof}
Let $p = \Top(B_j)$. Let $\lrarrow{\lrarrow{\lrarrow{p_1}{p_2}}{\dots}}{p_n}$ 
be the set of points added by $\GRE$ in $\RG(\BB_r)$
while processing $p$. By definition, only $p_1$ and $p_n$ can shift their status 
from key-not-live to key-live at time $\Time{p}+1$ as 
each other point $p_i$ is hidden by the point $p_{i-1}$ and $p_{i+1}$.
Let $q \in B_j$ such that $q \ne \Top(B_j)$. By Lemma \ref{lem:blockleftright}, all the points added by $\GRE$ 
while processing $q$ are either in $\BOX(B_j)$ or below $\Left(B_j)$ and $\Right(B_j)$.
Note that all the points that are added in $\BOX(B_j)$ are neither key-live or key-not-live 
as they do not lie in $\RG(\BB_r)$. 
So, they cannot change their status from key-not-live to key-live.
So after processing point $\MAXT(B_j)$, at most 4 keys $\Key{p_1}, \Key{p_n}, \Key{\Left(B_j)}, \Key{\Right(B_j)}$
can shift their status from key-live to key-not-live. This implies that 
$C_{\Time{\MAXT(B_j)}+1} \le C_{\Time{\MINT(B_j)}} + 4$.
\end{proof}

By Lemma \ref{lem:key-liveleftblocks}, if $B_j \in \BB_{\ell}$, $C_{\Time{\MAXT(B_j)}+1} \le C_{\Time{\Top(B_j)}} - (|R_{\Top(B_j)}|-2) + 2|\{B_k \in \BB_r | \Top(B_j) \in \LeftRel(B_k) \}| $.
By Lemma \ref{lem:key-liveleftblocks}, if $B_j \in \BB_{\ell\ell}$, $C_{\Time{\MAXT(B_j)}+1} \le C_{\Time{\Top(B_j)}} + 2|\{B_k \in \BB_r | \Top(B_j) \in \LeftRel(B_k) \}| $.
By Lemma \ref{lem:key-liverightblock}, if $B_j \in \BB_{rr}$, $C_{\Time{\MAXT(B_j)}+1} \le C_{\Time{\Top(B_j)}} + 2|\{B_k \in \BB_r | \Top(B_j) \in \RightRel(B_k) \}| $

The above two sums  along with the inequality in Lemma \ref{lem:key-livesameblock} are telescoping sums that starts at the top of the parent block B and ends at $\MAXT(B)$. Adding them gives the following expression.

$C_{\Time{\MAXT(B)}+1} - C_{\Time{\Top(B)}} \le -\sum_{B_j\in \BB_{\ell}} 
		(|R_{\Top(B_j)|}-2)+ \sum_{B_j \in \{\BB_{\ell},\BB_{\ell\ell}\}} 2|\{B_k \in \BB_r | \Top(B_j) \in \LeftRel(B_k	) \}|
			+ \sum_{B_j \in \BB_{rr}} 2|\{B_k \in \BB_r | \Top(B_j) \in \RightRel(B_k) \}|
			+ \sum_{B_j \in \BB_{r}} 4 $.
Since, there are no points in $\RG(\BB_r)$ at time $\Time{\Top(B)}$, $C_{\Time{\Top(B)}} = 0$.
Also there can at most be two points in $\LeftRel(B_k)$ and $\RightRel(B_k)$, so $2|\{B_k \in \BB_r | \Top(B_j) \in \LeftRel(B_k) \}| + 2|\{B_k \in \BB_r | \Top(B_j) \in \RightRel(B_k) \}| \le 4 \sum_{B_k \in \BB_r} 2 = 8m$. 
This leads to the following inequality:
$0 \le -\sum_{B_j\in B_{\ell}} (|R_{\Top(B_j)}|) + 8m + 4m$. So, $\sum_{B_j\in B_{\ell}} (|R_{\Top(B_j)}|) \le 12m$.
While processing $\Top(B_j)$, all the points added by $\GRE$ other than $R_{\Top(B_j)}$ are key-new. By Lemma \ref{lem:key-newinmfc}, 
these points are in $\MFC$. Note that even points in $R_{\Top(B_j)}$ can be in $\MFC$. However,
we have shown that number of such points is $\le 12m$. Thus we have proved Lemma \ref{lem:partitionlemma}.

\section{Upper Bounding $|\MFC|$: Good Pairs }
\label{sec:goodrectangles}

To prove that $\GRE$ has small competitive ratio we need to show that $|\GG|/|\XX \cup \OPT(\XX)|$ is small.
But our understanding of this ratio remains quite limited with the best upper bound being $O(\log n)$. 
One way to prove that $\GRE$ has small competitive ratio, sidestepping the above issue,
would be to prove a 
\emph{lower} bound on $\OPT(\XX)$ (the minimum cardinality point set 
that must be added to $\XX$ to make $(\XX \cup \OPT(\XX))$ arborally satisfied) by constructing a 
\emph{certificate} of the lower bound and then show that
the the ratio between $|\GG|$ and the lower bound is small. While working with the lower bound might
give a worse guarantee, but it might also allow more flexibility in the proof. This is a standard approach
in proving approximation ratios of algorithms. For the BST problem, there are many lower bounds known
(see \cite{DemianeHIKP09}). One lower bound, called the independent set lower bound from 
Demaine et al.~\cite{DemianeHIKP09} subsumes the previous ones. It is defined as follows:
rectangles $(p, q)$ and $(r, s)$ 
with $p,q,r,s \in \XX$, 
are \emph{independent} (in $\XX$) if they 
are not arborally satisfied and
no corner of either rectangle is strictly inside the other.
They showed that the cardinality of any independent set of rectangles 
provides a lower bound on $|\XX \cup \OPT(\XX)|$ as follows:

\begin{claim}[Claim 4.1, \cite{DemianeHIKP09}]
\label{claim:independentset}
Let $\XX$ contain an independent set of rectangles $I$ and 
let $\OPT(\XX)$ be a minimum cardinality point set 
that must be added to $\XX$ to make $(\XX \cup \OPT(\XX))$ arborally satisfied, then 
$|\XX \cup \OPT(\XX)| \ge |I|/2+|\XX|$.  
\end{claim} 

Though the above lemma gives a lower bound, it is not clear how 
to construct the set $I$ (lower bound certificate), or to relate it to $\GG$. 
Demaine et al. provided an alternative lower bound that is efficiently 
computable by a procedure called $\SGRE$ which is very similar to $\GRE$ and is within a 
constant factor of the best 
independent set lower bound. However, it is not known how to relate this lower bound to $|\GG|$, or
how to relate the executions of $\SGRE$ and $\GRE$ despite their close similarity. 

Our work provides a lower bound that can be related to $|\GG|$ on $k$-decomposable sequences. 
The lower bound certificate is constructed by directly looking at the execution of $\GRE$. Our construction
builds upon the idea of the independent set lower bound \cite{DemianeHIKP09}, but the final 
construction does not provide an independent set, but a more nuanced certificate. 
The above features of our technique give 
us hope that our techniques can be refined to better understand the performance of $\GRE$. 

\iflong
	In Sec.~\ref{sec:coupling}, Observation \ref{obs:pointinside} associated 
	each pair in $\MFC$ with a rectangle. 
\fi
The set of rectangles thus formed
(associated with the pairs in $\MFC$) 
is tightly coupled with the execution of $\GRE$. We
will first partition $\MFC$ into two parts: (1) $\GR$ and 
(2) $\BR$. While the set $\GR$ can be quite different from independent sets, we show that $\GR$ behaves like 
an independent set in the following sense: 

\iflong
	\begin{theorem}
	\label{lem:goodbound}
\else
	\newtheorem*{lem:goodbound}{Theorem \ref{lem:goodbound}}
	\begin{lem:goodbound}	
\fi
Let $\XX$ be the original point set and let  $\OPT(\XX)$ be the minimum number of points that 
must be added to $\XX$ to make it arborally satisfied. Then $|\XX \cup \OPT(\XX)| \ge \frac{|\GR|}{2} + |\XX|$.
\iflong
	\end{theorem}
\else	
	\end{lem:goodbound}
\fi


The rest of this section is devoted in proving Theorem \ref{lem:goodbound}.
\subsection{Good Pairs and their properties}
\label{subsubsec:properties}
We will now partition $\MFC$ into two parts:
\iflong
	\begin{definition} (Good and bad pairs)
	\label{recdef}	
\else
	\newtheorem*{recdef}{Definition \ref{recdef} (Good and bad pairs)}
	\begin{recdef}	
\fi
$\CP(p_1)=(p,q)$ with $\CP(p_1) \in \MFC$ is called a \emph{good} 
pair if $r \notin (^q\Box_p)^\circ$ (or $(_p\Box^q)^\circ$) 
for all $r \in \XX$. Otherwise,
$\CP(p_1)=(p,q)$ is called a \emph{bad} pair. 
Let $\GR := |\{p \in \GG \mid \RR(p)~ \text{is a good pair} \}|$ and
$\BR := |\{p \in \GG \mid \RR(p)~ \text{is a bad pair} \}|$.
\iflong
	\end{definition}
\else	
	\end{recdef}
\fi

Towards the end of proving Theorem~\ref{lem:goodbound}, we will show some properties of good pairs in 
the next lemma.
The definition of a good pair forbids any point from $\XX$ being in the interior of the rectangle. 
The following lemma proves that points from $\GG$ also cannot lie in the interior of the rectangle. 

\iflong
	\begin{lemma}
	\label{lem:property-good-rectangle}
\else
	\newtheorem*{lem:property-good-rectangle}{Lemma \ref{lem:property-good-rectangle}}
	\begin{lem:property-good-rectangle}	
\fi

Let $\RR(p_1) = {^p\Box_q}$ (or $_p\Box^q$) be a good pair. 
Then there are no points in $(^p\Box_q)^\circ$ (or $(_p\Box^q)^\circ$).
\iflong
	\end{lemma}
\else	
	\end{lem:property-good-rectangle}
\fi
\iflong
	\begin{figure}[H]

\centering
\begin{tikzpicture}

\draw[step=0.25,black,opacity=0.05]  grid (3,3);

\fill [ black] (2.5,0.5) circle (2pt) node[below,black,opacity=1] {\tiny{$p$}};
\fill [blue] (2,0.5) circle (2pt) node[below,black,opacity=1] {\tiny{$p_1$}};

\fill [ black] (0.5,2.5) circle (2pt) node[left,black,opacity=1] {\tiny{$q$}};
\fill [ blue] (2.,2.5) circle (2pt) node[above,black,opacity=1] {\tiny{$q_1$}};

\fill [ blue] (1.25,2) circle (2pt) node[above,black,opacity=1] {\tiny{$s$}};
\fill [ blue] (1.25,2.75) circle (2pt) node[above,black,opacity=1] {\tiny{$t$}};

\fill [ black] (0.25,2) circle (2pt) node[below,black,opacity=1] {\tiny{$\OP(s)$}};

\fill [gray, opacity=0.2] (0.5,2.5) rectangle (2,0.5);

\end{tikzpicture}
\caption{ Illustration of the case dealt in the proof of Lemma \ref{lem:property-good-rectangle} 
when $s$ lies in $^q\Box_{p_1}$.}
\label{fig:good-rectangle}
\end{figure}
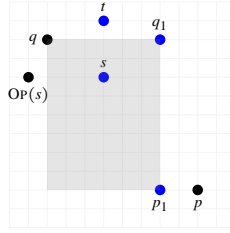
	\begin{proof}
	By symmetry, it suffices to prove the lemma for rectangles of type $^q\Box_p$.
	Note that no original point can be in $(^q\Box_p)^{\circ}$ as $\RR(p_1)$ is a good pair.
	Let $q_1$ be the first point above $p_1$. This implies that
	$\GRE$ marks point $p_1$ to the left of $p$ due to the unsatisfied rectangle $^{q_1}\Box_p$.

	If $q_1 = \OP(q_1) (= q)$, then there are no points in 
	$(^{q}\Box_p)^{\circ}$ as $\GRE$ found $^{q_1}\Box_p$ arborally 
	unsatisfied. So let us assume that 
	$\lrarrow{q}{q_1}$.

	If there is a point $s \in (^{q}\Box_p)^{\circ}$ then it must satisfy one of the following three cases 
	(recall that we must have $s \in \GG$):

	\begin{enumerate}
	\item $s \in (^{q_1}\Box_p)^\circ$.

	Again, there exists no point in 
	$(^{q_1}\Box_p)^{\circ}$ as $\GRE$ found $^{q_1}\Box_p$ arborally 
	unsatisfied. So this case is not possible. 

	\item $s$ lies above $p_1$ in $(^q\Box_p)^{\circ}$.

	By the definition of $\CP(\cdot)$, $q_1$ is the first point above $p_1$. So, such
	a point $s$ cannot exist.

	\item $s \in (^q\Box_{p_1})^{\circ}$ (see Fig. \ref{fig:good-rectangle}).

	Let us assume that 
	$s$ is the closest point to $q$ in $(^q\Box_{p_1})^{\circ}$. 
	Therefore, 
	no point lies above $s$ in $(^q\Box_{p_1})^{\circ}$. However, there exists a point above
	$s$ that does not lie in $(^q\Box_{p_1})^{\circ}$. $\UP(s)$ is one such candidate.
	Let $t$ be the first point above $s$. Note that $t$ can lie to the right or 
	north-east of $q$ (though in Fig. \ref{fig:good-rectangle} it lies to the north-east of $q$).

	Note that $q_1$ hides $t$ from $p$ (as either $q_1$ lie to the right of $t$ in $^t\Box_p$
	 or $q_1 \in (^t\Box_p)^{\circ}$) and 
	$s$ is the first point below $t$ such that $_s\swnearrow^{q_1}$. 
	Then, by Lemma \ref{lem:hidden}, $_{\OP(s)}\swnearrow^{q_1}$.
	$\OP(s)$ cannot lie in $(^q\Box_{p_1})^{\circ}$ as $^q\Box_p$ is 
	a good pair; this implies $_{\OP(s)}\swnearrow^q$.
	In that case, $q$ hides $t$ from $\OP(s)$ (as either $q$ lies to the left of $t$ in $_{\OP(s)}\Box^t$ 
	or $q \in (_{\OP(s)}\Box^t)^{\circ}$)
	and $s$ is the first point below $t$
	such that $^q\senwarrow_s$. 
	Again by Lemma \ref{lem:hidden-right}, $^q\senwarrow_{\OP(s)}$. 
	This leads to a contradiction as we have already deduced that $_{\OP(s)}\swnearrow^q$.
	So our assumption
	that $s \in (^q\Box_{p_1})^{\circ}$  must be false.
	\end{enumerate}
	\end{proof}
\fi
\subsection{Interaction between Good Pairs}
\label{subsubsec:interaction}
The definition of {\em independent sets} specifies 
how two rectangles can intersect each other. The set of good pairs need not be independent in general. 
However, there are constraints on the way good pairs interact. 
Specifically, we will show that if two good pairs intersect
then the point associated with one of them cannot lie in the interior 
of the rectangle associated with the other pair:
\iflong
	\begin{lemma}
	\label{lem:interthree}
\else
	\newtheorem*{lem:interthree}{Lemma \ref{lem:interthree}}
	\begin{lem:interthree}	
\fi

Let $\RR(q_1) = {^s\Box_q}$ and $\RR(q_2) = {^p\Box_q}$ be two good pairs  such that $_s\neswarrow^p$.
Assume that the intersection between the two rectangles is of type $\interthree$. 
Then $q_1$ cannot lie to the left of $q$ in $^p\Box_q$ except at the 
bottom-left corner
of $^p\Box_q$.
\iflong
	\end{lemma}
\else	
	\end{lem:interthree}
\fi
\iflong
	\begin{figure}[H]

\centering
\begin{tikzpicture}

\draw[step=0.25,black,opacity=0.2] (1,0.25) grid (3,2.5);

\draw [black,opacity=0.2] (1,0.25) -- (1,2.5);
\fill [ black] (1.75,2) circle (2pt) node[above,black,opacity=1] {\tiny{$p$}};
\fill [ black] (1.25,1.5) circle (2pt) node[above,black,opacity=1] {\tiny{$s$}};

\fill [ black] (2.75,0.5) circle (2pt) node[below,black,opacity=1] {\tiny{$q$}};

\fill [ blue] (2,0.5) circle (2pt) node[below,black,opacity=1] {\tiny{$q_1$}};
\fill [ blue] (2.25,0.5) circle (2pt) node[below,black,opacity=1] {\tiny{$q_2$}};

\fill [ blue] (2,1.5) circle (2pt) node[above,black,opacity=1] {\tiny{$s_1$}};

\fill [gray, opacity=0.2] (1.75,2) rectangle (2.75,0.5);
\fill [gray, opacity=0.2] (1.25,1.5) rectangle (2.75,0.5);
\end{tikzpicture}

\caption{ Illustration of the bad case in the proof of \ref{lem:interthree} when $q_1$ lies in 
to the left of $q$ in $^p\Box_q$ (but not at its bottom left corner). }
\label{fig:interthree}
\end{figure}
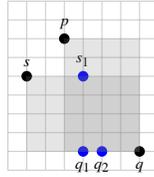
	\begin{proof}
	The intersection of $^p\Box_q$ and $^s\Box_q$ is of type $\interthree$.
	Let $s_1$ be the first point above $q_1$. If $q_1$ lies 
	to the left of $q$ in $^p\Box_q$ (but not at the 
	bottom-left corner of $^p\Box_q$), 
	then it lies to the left or right of $q_2$ in $^p\Box_q$.
	This implies that $s_1$ lies in $(^p\Box_q)^{\circ}$.
	However, Lemma \ref{lem:property-good-rectangle} forbids such a situation.

	\end{proof}
\fi

\iflong
	\begin{lemma}
	\label{lem:interfour}
\else
	\newtheorem*{lem:interfour}{Lemma \ref{lem:interfour}}
	\begin{lem:interfour}	
\fi
Let $\RR(q_1) = {^s\Box_q}$ and $\RR(r_1) = {^s\Box_r}$ be two good pairs  such that $_r\neswarrow^q$.
Assume that the intersection between the two rectangles is of type $\interfour$.
Then $q_1 \notin (^s\Box_r)^{\circ}$ and $q_1$ cannot be at the bottom-left corner of $^s\Box_q$.
\iflong
	\end{lemma}
\else	
	\end{lem:interfour}
\fi

\iflong
	\begin{figure}[H]

\centering
\begin{tikzpicture}

\draw[step=0.25,black,opacity=0.2] (1,0.25) grid (3,2.5);

\draw [black,opacity=0.2] (1,0.25) -- (1,2.5);
\fill [ black] (2.25,0.5) circle (2pt) node[below,black,opacity=1] {\tiny{$r$}};
\fill [ black] (1.5,2) circle (2pt) node[above,black,opacity=1] {\tiny{$s$}};

\fill [ black] (2.75,1) circle (2pt) node[below,black,opacity=1] {\tiny{$q$}};

\fill [ blue] (1.5,1) circle (2pt) node[below,black,opacity=1] {\tiny{$q_1$}};

\fill [ blue] (1.75,0.5) circle (2pt) node[below,black,opacity=1] {\tiny{$r_1$}};
\fill [ blue] (1.75,2) circle (2pt) node[above,black,opacity=1] {\tiny{$s_1$}};

\fill [gray, opacity=0.2] (1.5,2) rectangle (2.75,1);
\fill [gray, opacity=0.2] (1.5,2) rectangle (2.25,0.5);
\end{tikzpicture}

\caption{Illustration of the bad case in the proof of \ref{lem:interfour}
when $q_1$ lies at the bottom-left corner of $^s\Box_q$.}
\label{fig:interthree}
\end{figure}
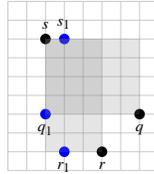
	\begin{proof}
	By Lemma \ref{lem:property-good-rectangle}, $q_1 \notin (^s\Box_r)^{\circ}$.\

	Assume for contradiction that $q_1$ lies at the bottom-left corner of $^s\Box_q$, 
	that is, it is the first point below $s$ (as $\CP(q_1) = (q,s)$). This implies that 
	$r_1$ cannot lie below $q_1$ (as then $\CP(r_1) \ne (r,s))$.
	So $^{q_1}\senwarrow_{r_1}$. Let $s_1$ be the first point above $r_1$.
	Since $\CP(r_1) = (r,s)$, $s_1$ lies to the right of $s$. 
	This implies that $s_1$ hides $s$ from $q$ (as $s_1$ lies to the right
	of $s$ in $^s\Box_q$) and $q_1$ is the first point below $s$. 
	By Lemma \ref{lem:hidden}, $_{\OP(q_1)}\neswarrow^{s_1}$.
	But $\OP(q_1)=q$ lies to the south-east of $s_1$. 
	So, we arrive at a contradiction. So 
	$q_1$ cannot lie at the bottom-left corner of $^s\Box_q$

	\end{proof}

The following two lemmas are also a direct consequence of  Lemma \ref{lem:property-good-rectangle}.
\fi

\iflong
	\begin{lemma}
	\label{lem:interone}
\else
	\newtheorem*{lem:interone}{Lemma \ref{lem:interone}}
	\begin{lem:interone}	
\fi
Let $\RR(q_1) = {^s\Box_q}$ and $\RR(r_1) = {^p\Box_r}$ be two good pairs  such that $_p\neswarrow^s$,$_r\neswarrow^q$.
Assume that the intersection between the two rectangles is of type $\interone$.
Then $q_1 \notin (^p\Box_r)^{\circ}$.
\iflong
	\end{lemma}
\else	
	\end{lem:interone}
\fi

\iflong
	\begin{lemma}
	\label{lem:intertwo}
\else
	\newtheorem*{lem:intertwo}{Lemma \ref{lem:intertwo}}
	\begin{lem:intertwo}	
\fi
Let $\RR(q_1) = {^s\Box_q}$ and $\RR(r_1) = {^p\Box_r}$ be two good pairs such that $_s\neswarrow^p$,$_r\neswarrow^q$.
Assume that the intersection between the two rectangles is of type $\intertwo$.
Then $q_1 \notin (^p\Box_r)^{\circ}$.
\iflong
	\end{lemma}
\else	
	\end{lem:intertwo}
\fi

\subsection{Putting it together}
\label{subsubsec:finaldetail}

Let $\GR_{\boxbackslash} := \{^q\Box_p \mid {^q\Box_p} \in \GR \}$. Similarly we can define
$\GR_{\boxslash}$. 
Demaine et al.~\cite{DemianeHIKP09} proved that the size of an {\em independent set} of rectangles
provides a lower bound on $|\OPT(\XX)|$ (Claim~\ref{claim:independentset}).
Instead of independent sets we have to argue about $\GR$. Our argument would follow that of Demaine et al. at a high
level, but with some important changes. 
We state three lemmas below which are adaptations of lemmas from \cite{DemianeHIKP09} to $\GR$.
\iflong
	\begin{lemma} (compare Lemma 4.4, \cite{DemianeHIKP09})
	\label{lem:demiane1}
\else
	\newtheorem*{lem:demiane1}{Lemma \ref{lem:demiane1} (compare Lemma 4.4, \cite{DemianeHIKP09})}
	\begin{lem:demiane1}	
\fi
Let $q$ be the point in $\XX$ with the maximum $x$-coordinate such that there 
exists a rectangle with $q$ in $\GR_{\boxbackslash}$ as one of its diagonal point.
Let  $^s\Box_q \in \GR_{\boxbackslash}$ be the widest rectangle with $q$ as one of its diagonal point. 
Then we can find a vertical line $\ell$ passing through the interior of $^s\Box_q$, 
such that inside $^s\Box_q$, $\ell$ does not intersect the interior 
of any other rectangle in $\GR_{\boxbackslash} \setminus \{^s\Box_q \}$.  
\iflong
	\end{lemma}
\else	
	\end{lem:demiane1}
\fi
\iflong
	\begin{proof}
	Let $\CP(q_1) =~ ^s\Box_q$. 
	Let $S$ denote the set of all rectangles in 
	$\GR_{\boxbackslash}\setminus \{^s\Box_q\}$ overlapping $^s\Box_q$.
	Let us assume for contradiction that $S$ spans the horizontal section of 
	$^s\Box_q$.
	Let $^p\Box_r$ be a rectangle in $S$ with $\CP(r_1) =~ ^p\Box_r$. 

	Note that any intersection between $^s\Box_q$ and $^p\Box_r$
	can be of type $\interone,\intertwo,\interthree$ or $\interfour$.
	This is due to the fact that any other intersection 
	will either force two points from $\{p,q,r,s\}$ to be on the same horizontal or 
	vertical line (which is forbidden as these points come from a permutation sequence),
	or it will force a diagonal point (in $\XX$) of one rectangle 
	to be in the interior of the other.
	Also, by the construction of $^s\Box_q$, $^p\Box_r$ is the left rectangle
	in the intersection type $\interone$ and the narrower rectangle 
	in intersection type $\intertwo,\interthree$ or $\interfour$.

	The right edge of any rectangle in $\GR_{\boxbackslash}$
	does not intersect even partially with the left edge of any other 
	rectangle (as points in $\XX$ come from a permutation sequence).
	So, if the rectangles in $S$ span the horizontal section of $^s\Box_q$, 
	then each boundary point to the left of $q$ in $^s\Box_q$ (except maybe its bottom-left corner)
	lies either in the interior of a rectangle $^p\Box_r$ (intersection type $\interone,\intertwo,\interfour$)
	or to the left of $q$ in $^p\Box_q$ except at the bottom-left corner of $^p\Box_q$
	(intersection type $\interthree$ and $r=q$ ). 

	Then by Lemmas \ref{lem:interthree}, \ref{lem:interfour}, \ref{lem:interone}, and \ref{lem:intertwo}, point
	$q_1$ cannot lie to the left of $q$ in $^s\Box_q$,
	except possibly its bottom-left corner. 

	As rectangles of intersection type $\intertwo$ and $\interthree$ 
	are narrower than $^s\Box_q$ (by the construction of $^s\Box_q$), 
	these rectangles cannot span horizontal section of $^s\Box_q$. 
	Since we assumed that rectangles in $S$ span the horizontal section of $^s\Box_q$, 
	it implies that $S$ contains a rectangle of type $\interone$ or $\interfour$.
	Using Lemmas  \ref{lem:interone} or \ref{lem:interfour} respectively, 
	$q_1$ cannot even lie at the bottom-left corner of $^s\Box_q$.
	This implies that $q_1$ does not lie to the left of $q$ in $^s\Box_q$.
	However, this contradicts Observation \ref{obs:pointinside} that states  
	if $\CP(q_1) =~ ^s\Box_q$, then $q_1$ lies to the left of $q$ in $^s\Box_q$.
	So we arrive at a contradiction, hence our assumption that ``$S$ spans the horizontal section of 
	$^s\Box_q$" must be false. So, we can find a vertical line $\ell$ through the interior $^s\Box_q$, 
	such that inside $^s\Box_q$, $\ell$ does not intersect the interior 
	of any other rectangle in $\GR_{\boxbackslash} \setminus \{^s\Box_q \}$.  

	\end{proof}
\fi

\iflong
	\begin{lemma} (compare Lemma 4.3, \cite{DemianeHIKP09})
	\label{lem:demiane2}
\else
	\newtheorem*{lem:demiane2}{Lemma \ref{lem:demiane2} (compare Lemma 4.3, \cite{DemianeHIKP09})}
	\begin{lem:demiane2}	
\fi
Given $^p\Box_q \in \GR_{\boxbackslash}$  
and a vertical line $\ell$ at a non-integer $x$-coordinate intersecting $(^p\Box_q)^{\circ}$ and 
a set of points $\YY$ such that each pair of points in $\XX \cup \YY$ is arborally satisfied.
We can find two points $a,b \in (\XX \cup \YY)$
in $^p\Box_q$ such that $\lrarrow{a}{b}$, $a$ is to the left of $\ell$, $b$
is to the right of $\ell$, and there are no points in $\XX \cup \YY$ on the horizontal
segment connecting $a$ to $b$.
\iflong
	\end{lemma}
\else	
	\end{lem:demiane2}
\fi
\iflong
\begin{proof}
	Let $a$ and $b$  in $(\XX \cup \YY)$ be two closest points in $^p\Box_q$ such 
	that $a$ is the left of line $\ell$ and $b$ is to the right of $\ell$.
	Note that $a$ exists ($p$ is one such candidate). Similarly $b$
	exists ($q$ is one such candidate).
	By construction,
	there are no points in $^a\Box_b$ (or $_a\Box^b$).
	If $a$ and $b$ do not lie on the same horizontal line, then 
	$^a\Box_b$ (or $_a\Box^b$) is an arborally unsatisfied rectangle,
	contradicting our assumption that $(\XX \cup \YY)$ is an arborally satisfied set. 
	\end{proof}
\fi

\iflong
	\begin{lemma} (compare Lemma 4.5, \cite{DemianeHIKP09})
	\label{lem:demiane3}
\else
	\newtheorem*{lem:demiane3}{Lemma \ref{lem:demiane3} (compare Lemma 4.5, \cite{DemianeHIKP09})}
	\begin{lem:demiane3}	
\fi
Given a point set $\XX$ and another point set $\YY$ such that $\XX \cup \YY$ is 
arborally satisfied, then $|\XX \cup \YY| \ge |\GR_{\boxbackslash}| + |\XX|$.
\iflong
	\end{lemma}
\else	
	\end{lem:demiane3}
\fi
\iflong
	\begin{proof} This proof is essentially verbatim from \cite{DemianeHIKP09}. The property of 
	independent set of rectangles used in the proof of Lemma 4.5 in \cite{DemianeHIKP09} is the
	existence of the line $\ell$ (proved in Lemma 4.4 there). Our Lemma~\ref{lem:demiane1} proves
	the existence of line $\ell$ for $\GR_{\boxbackslash}$. Given this, the proofs for independent set 
	and for $\GR_{\boxbackslash}$ are the same. 

	We apply Lemma \ref{lem:demiane1} to find a rectangle $^s\Box_q$ in $\GR_{\boxbackslash}$ and
	a vertical line $\ell$ piercing $^s\Box_q$ with the property that no other
	rectangle in $\GR_{\boxbackslash}$ 
	intersects $\ell$ in the interior of $^s\Box_q$. Then we apply
	Lemma~\ref{lem:demiane2} to find two points $a,b$ horizontally adjacent in 
	$\XX \cup \YY$
	and on opposite sides of $\ell$ in $^s\Box_q$. We mark this pair $(a,b)$
	with rectangle $^s\Box_q$. Then we remove 
	$^s\Box_q$  from $\GR_{\boxbackslash}$ and repeat
	the process, until there are no rectangles left in $\GR_{\boxbackslash}$.
	Whenever we remove a rectangle $^s\Box_q$ 
	 from $\GR_{\boxbackslash}$, if $a$ and
	$b$ are not on the top or bottom sides of $^s\Box_q$, then $a$ and $b$
	do not simultaneously belong to any other rectangle in 
	$\GR_{\boxbackslash} \setminus \{^s\Box_q \}$,
	so they will never be marked again. On the other hand, if
	$a$ and $b$ are on the top (bottom) side of $^s\Box_q$, then $a$ and $b$
	are neither in the interior nor on the top (bottom) side of any other
	rectangle in $\GR_{\boxbackslash}$. Furthermore, since coordinates in $\XX$ are distinct, the top side
	of no rectangle in $\GR_{\boxbackslash}$ coincides even partially with the bottom
	side of a rectangle in $\GR_{\boxbackslash}$. Thus, each pair of horizontally
	adjacent points in $\XX \cup \YY$ can be marked at most once.
	Finally, by distinctness of $y$-coordinates in $\XX$, at most
	one point in a pair  can
	belong to $\XX$. Therefore the number of points in $\YY$ is at
	least $|\GR_{\boxbackslash}|$, proving the lemma.

	\end{proof}
\fi

Similar to Lemma \ref{lem:demiane3}, we can also show that 
any arborally satisfied set $\YY$ satisfies $|\XX \cup \YY| \ge |\GR_{\boxslash}| + |\XX|$.
So $|\XX \cup \YY| \ge \frac{|\GR_{\boxbackslash}| + |\GR_{\boxslash}|}{2} + |\XX| = \frac{|\GR|}{2} + |\XX|$.
We have thus proved Theorem~\ref{lem:goodbound}.

\section{Upper bounding $|\MFC|$: Bad Pairs}
\label{sec:mfc}
In this section we prove 
\begin{lemma}\label{lem:bad-mfc-bound}
$|\BR| \le  10 n(k-1)$.
\end{lemma}

For a point $p \in \XX$, similar to $\RMMC_p$, define
$\RMFC_p := \{ \CP(p_1) \in \BR \mid \lrarrow{p}{p_1} \text{ and } \CP(p_1) = (p,q),  
\text{ where } q \in \XX \}$. Let $\CP(p_1) \in \RMFC_p$. 
To upper bound $|\BR|$, we construct $\AMmc(\cdot)$ that maps a point $u$ (with $\CP \in \BR$)
to a point $v$ (with $\CP(v) \in \MMC$). 
$\AMmc(\cdot)$ takes at most $4$ points (with $\CP(\cdot) \in \BR$) 
to a point (with $\CP(\cdot) \in \MMC$), and we already proved $|\MMC| = O(nk)$. 
This would provide an upper bound on $|\BR|$. Unfortunately,
$\AMmc(\cdot)$ is a partial map, and does not map every point (with $\CP(\cdot) \in \BR$). We show, by an argument 
similar to the
one used for bounding $|\MMC|$, that the number of unmapped points is $O(nk)$. This gives our desired
bound $|\BR| = O(nk)$. 

The rest of this section is devoted to proving Lemma~\ref{lem:bad-mfc-bound}. 
\iflong
	In Sec.~\ref{sec:construction-map} 
	we construct
	the map mentioned above; 
	in Sections~\ref{sec:bound-points-mapped}, \ref{sec:bound-points-not-mapped} we provide upper bounds on the size of the sets
	of mapped and unmapped points; 
	finally, we put these two bounds together to complete the proof. 
\fi







\subsection{Construction of $\AMmc(\cdot)$ and its properties}
\label{sec:construction-map}

\iflong
We begin by describing the setting used in the proof.
\fi 
\paragraph{The setting used in the proof.}
Unless mentioned otherwise, we will assume in this section that $\RT(B)=p$ and $\CP(p_1) \in \RMFC_p$.
Let $q_1$ be the first point above $p_1$. 
Since $\GRE$ has put a point $p_1$ below $q_1$ while processing $p$, 
it found $_p\Box^{q_1}$ arborally unsatisfied ($\lrarrow{p}{p_1}$ as $\CP(p_1) \in \RMFC_p$).
This implies that there are no points in $_p\Box^{q_1}$ except $p,q_1$ and $p_1$.
If $\OP(q_1)=q_1$, then 
$\RR(p_1) \in \GR$. So let us assume that $\OP(q_1)=q$ and $\lrarrow{q_1}{q}$ 
(since $\RR(p_1) \in \RMFC_p$ we do not have $\lrarrow{q}{q_1}$). 

Since $\RR(p_1) \in \BR$, there
is another point in $\XX \setminus \{p,q\}$ and in $(_p\Box^q)^{\circ}$. 
Let $r \in \XX$ be a point in $(_p\Box^q)^{\circ}$ having the smallest $x$-distance from $p$.
Let $p_1'$ be the next point to the right of $p_1$; that $p_1'$ exists is proven next:

\iflong
	\begin{lemma}
	\label{lem:p_2exists}
\else
	\newtheorem*{lem:p_2exists}{Lemma \ref{lem:p_2exists}}
	\begin{lem:p_2exists}	
\fi
There is a point $p_1'$ to the right of $p$  in $_{p_1}\Box^r$
but not at the bottom-right corner of $_{p_1}\Box^r$.
\iflong
	\end{lemma}
\else	
	\end{lem:p_2exists}
\fi
\iflong
	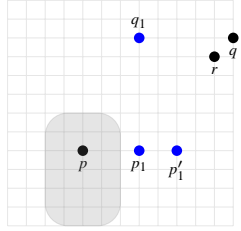
\begin{figure}[H]
\centering
\begin{tikzpicture}
[round/.style={rectangle, rounded corners=3mm, minimum size=20mm, draw=black!50}]

\draw[step=0.25,black,opacity=0.1] (0,0) grid (3,3);

\fill [black] (1,1) circle (2pt) node[below,black,opacity=1] {\tiny{$p$}};

\fill [ blue] (1.75,1) circle (2pt) node[below,black,opacity=1] {\tiny{$p_1$}};
\fill [ blue] (2.25,1) circle (2pt) node[below,black,opacity=1] {\tiny{$p_1'$}};
\fill [ black] (2.75,2.25) circle (2pt) node[below,black,opacity=1] {\tiny{$r$}};

\fill [ blue] (1.75,2.5) circle (2pt) node[above,black,opacity=1] {\tiny{$q_1$}};
\fill [ black] (3,2.5) circle (2pt) node[below,black,opacity=1] {\tiny{$q$}};

\draw [round,fill=gray,opacity=0.2] (0.5,0) rectangle (1.5,1.5);

\end{tikzpicture}
\caption{Lemma \ref{lem:p_2exists} states that $p_1'$ exists to the right of $p_1$ such that $_{p_1'}\swnearrow^r$ }
\label{fig:mfc-setting}
\end{figure}
	\begin{proof}
	Consider the rectangle $^{q_1}\Box_r$. By Lemma \ref{lem:greedy-property-left},
	a point $r'$ lies either to the left or above $r$ in $^{q_1}\Box_r$.
	Since $r \in \XX$, by Observation \ref{obs:nothing-above}, $\GRE$ does not put any point above $r$. Also, since $q_1$ is 
	the first point above $p_1$, $r'$ cannot lie on the bottom-left endpoint of
	$^{q_1}\Box_r$. So $\lrarrow{r'}{r}$ and $_{p_1}\neswarrow^{r'}$.

	Now consider rectangle $_{p_1}\Box^{r'}$. Again by Lemma \ref{lem:greedy-property}, 
	a point $p_1'$ must exist either to the right or above $p_1$ in $_{p_1}\Box^{r'}$.
	Since $q_1$ is the first point above $p_1$ (and $q_1$ does not lie in $_{p_1}\Box^{r'}$),
	$p_1'$ cannot lie above $p_1$ in $_{p_1}\Box^{r'}$. So it must lie to the right of $p_1$ in $_{p_1}\Box^{r'}$.

	As $\lrarrow{r'}{r}$, we deduce that $p_1'$ cannot lie below $r$.
	\end{proof}
\fi



We are now ready to describe the construction of the partial map 
$\AMmc(\cdot)$. Its domain is $\{u \in \GG \mid \CP(u) \in \BR\}$ and its range is  
$\{v \in \GG \mid \RR(v) \in \MMC\}$.
$\AMmc(\cdot)$ is computed by the procedure in Fig.~\ref{fig:associateMMC}.
\iflong
	\begin{figure}[hpt]
\else
	\begin{figure}[!hpt]
\fi
\begin{mdframed}[style=MyFrame]
\begin{procedure}[H]
	\scriptsize
	 \If{$\CP(p_1') \in \MMC$ }{
		 $\AMmc(p_1) = p_1'$\;
	}
	\Else{
		Let $s_1'$ be the  first point above $p_1'$\;
		\If{$\CP(s_1') \in \MMC$}
		{	
			$\AMmc(p_1) = s_1'$\; 
		}
		\Else{
			$\AMmc(p_1)$ is not defined\;
			}
		}
\end{procedure}
\end{mdframed}
\caption{Procedure to compute $\AMmc(p_1)$ mapping $p_1$ with $\CP(p_1) \in \RMFC$ to a point 
with $\CP(\cdot) \in \MMC$. If $\CP(p_1) \in \LMFC$,
the procedure is symmetric.}
\label{fig:associateMMC}
\end{figure}

In the procedure for computing $\AMmc(p_1)$, if we find that $\CP(p_1') \in \MFC$, then we look for $s'_1$, the first
point above $p'_1$. 
\iflong
	We now prove some properties of $s_1'$; these will be used in Sec.~\ref{sec:bound-points-not-mapped}. 
	By Lemma~\ref{lem:p_2exists}, $p_1'$ cannot lie below $r$.
	So $_{p_1'}\neswarrow^r$ (see Fig.~\ref{fig:s_2-mfc}).
	Note that $_{s_1'}\neswarrow^r$ or $^{s_1'}\nwsearrow_r$ or $\lrarrow{s_1'}{r}$ 
	(though in Fig.~\ref{fig:s_2-mfc} we have shown the case $_{s_1'}\neswarrow^r$). 
	Since $\GRE$ has put a point below $s_1'$ while processing $p$, hence $_p\Box^{s_1'}$
	is an arborally unsatisfied rectangle.
	So, $^{q_1}\nwsearrow_{s_1'}$ as $\lrarrow{p_1}{p_1'}$ and 
	$\GRE$ marked these points due to unsatisfied rectangle $_p\Box^{q_1}$ and
	$_p\Box^{s_1'}$ respectively.  
	We make some other observation regarding $s_1'$.

	\begin{observation}
	\label{obs:s2obs1}
	There are no points to the left of $s_1'$ in $_p\Box^{s_1'}$.
	\end{observation}
	The above observation is a direct consequence of the fact that while
	processing $p$, $\GRE$ found $_p\Box^{s_1'}$  arborally unsatisfied 
	while processing $p$.

	\begin{observation} 
	\label{obs:s2obs2}

	There are no points in (1) $(^{q_1}\Box_{s_1'})^{\circ}$ and 
	(2) below $q_1$ or to the left of $s_1'$ in $^{q_1}\Box_{s_1'}$.
	\end{observation}

	Again this is due to the facts that (1) $p_1$ and $p_1'$ are two consecutive points
	to the right of $p$, (2) $\GRE$ found $_p\Box^{q_1}$ and $_p\Box^{s_1'}$ to be
	arborally unsatisfied while processing $p$, (3) $q_1$ is the first point above $p_1$. 

	\begin{observation}
	\label{obs:s2obs3}
	$s_1' \notin \XX$ and $\lrarrow{s_1'}{\OP(s_1')}$.
	\end{observation}
	Since $r$ is an original point in 
	$_p\Box^q$ with the smallest $x$-distance to $p$, $s_1' \notin \XX$.
	Together with our assumption that $\CP(p_1') \in \MFC$,
	this implies $\lrarrow{s_1'}{\OP(s_1')}$. 

	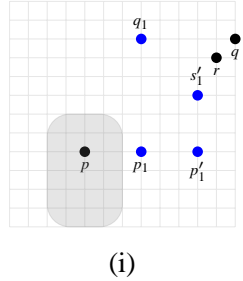
\begin{figure}[H]
\centering
\begin{tikzpicture}
[round/.style={rectangle, rounded corners=3mm, minimum size=20mm, draw=black!50}]

\draw[step=0.25,black,opacity=0.1] (0,0) grid (3,3);

\fill [black] (1,1) circle (2pt) node[below,black,opacity=1] {\tiny{$p$}};

\fill [ blue] (1.75,1) circle (2pt) node[below,black,opacity=1] {\tiny{$p_1$}};
\fill [ blue] (2.5,1) circle (2pt) node[below,black,opacity=1] {\tiny{$p_1'$}};
\fill [ black] (2.75,2.25) circle (2pt) node[below,black,opacity=1] {\tiny{$r$}};

\fill [ blue] (1.75,2.5) circle (2pt) node[above,black,opacity=1] {\tiny{$q_1$}};
\fill [ blue] (2.5,1.75) circle (2pt) node[above,black,opacity=1] {\tiny{$s_1'$}};
\fill [ black] (3,2.5) circle (2pt) node[below,black,opacity=1] {\tiny{$q$}};

\draw [round,fill=gray,opacity=0.2] (0.5,0) rectangle (1.5,1.5);
\node at (1.5,-0.5) {(i)};

\end{tikzpicture}
\caption{In $\AMmc(p_1)$, $s_1'$ is the first point above $p_1'$}
\label{fig:s_2-mfc}
\end{figure}



	\subsection{Bounding the number of points mapped by $\AMmc(\cdot)$}
	\label{sec:bound-points-mapped}
	
	Let $\CP(p_i) \in \RMFC_p$. 
	For a $q_i \in \GG$, 
	if $\AMmc(p_i)=q_i$, then $q_i \in \GG$, and either $q_i$ is the first point to the right of
	$p_i$ or $q_i$ is the first point above $q_i'$ where $q_i'$ is the first point 
	to the right of $p_i$. Hence for $q_i \in \GG$, there are at most two points $p_i$ such that 
	$\AMmc(p_i)=q_i$.
	In other words, for each $q_i \in \GG$, 
	we have $|\{ p_i \mid \AMmc(p_i) = q_i$ and $\RR(q_i) \in \MMC \}| \le 2$.
	So,

	 
	\begin{align*}
	\sum_{p \in \XX} |\{ \CP(p_i) \in \RMFC_p \mid p_i \in \mathrm{domain}(\AMmc) \}| 
	&=\sum_{q_i \in \GG } | \{ p_i \mid p_i \text{ s.t. } \AMmc(p_i) = q_i \text{ and } \RR(q_i) \in \MMC\}| \\
	&\le 2|\MMC|.
	\end{align*}

	By symmetry,\\
	\begin{align*}
	\sum_{p \in \XX} |\{ \CP(p_i) \in \LMFC_p \mid p_i \in \mathrm{domain}(\AMmc) \}| 
	&\le 2|\MMC|.
	\end{align*}
	Combining the above two inequalities we get
	\begin{lemma}
	$\sum_{p \in \XX} |\{ \CP(p_i) \in \BR_p \mid p_i \in \mathrm{domain}(\AMmc) \}| \le 4|\MMC|$.
	\end{lemma}
\fi

\subsection{Bounding the number of points not mapped by $\AMmc(\cdot)$}
\label{sec:bound-points-not-mapped}
\iflong
	We continue with the setting introduced in Sec.~\ref{sec:construction-map}.
	Since $\AMmc(\cdot)$ was not able to map $p_1$, the following must be true:
	\begin{enumerate}
	\item $\CP(p_1') \in \MFC$,
	\item $\CP(s_1') \in \MFC$.
	\end{enumerate}

	Since $\RR(p_1) \in \BR$, there exists another point $r \in \XX \setminus\{p,q\}$ in 
	$_p\Box^q$. 

	Let $t_1'$ be the first point above $s_1'$. So $\CP(s_1') = (\OP(s_1'),\OP(t_1'))$. 
	Since $\CP(s_1') \in \MFC$ and $\lrarrow{s_1'}{\OP(s_1')}$ 
	(Observation \ref{obs:s2obs3}), either $\lrarrow{\OP(t_1')}{t'_1}$ 
	or $\OP(t_1')=t_1'$. 

	\begin{observation}
	\label{obs:t2obs1}
	Either $\lrarrow{\OP(t_1')}{t'_1}$ or $\OP(t_1')=t_1'$.  
	\end{observation}

	We first make some claims on the position of $t_1'$.

	\begin{lemma}
	\label{lem:t2}
	Let $t_1'$ be the first point above $s_1'$. Then
	$^{q_1}\nwsearrow_{t_1'}$.
	\end{lemma}
	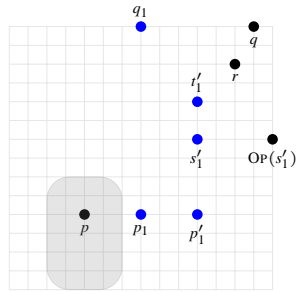
\begin{figure}[H]
\centering
\begin{tikzpicture}
[round/.style={rectangle, rounded corners=3mm, minimum size=20mm, draw=black!50}]

\draw[step=0.25,black,opacity=0.1] (0,0) grid (3.5,3.5);

\fill [black] (1,1) circle (2pt) node[below,black,opacity=1] {\tiny{$p$}};

\fill [ blue] (1.75,1) circle (2pt) node[below,black,opacity=1] {\tiny{$p_1$}};
\fill [ blue] (2.5,1) circle (2pt) node[below,black,opacity=1] {\tiny{$p_1'$}};
\fill [ black] (3,3) circle (2pt) node[below,black,opacity=1] {\tiny{$r$}};

\fill [ blue] (1.75,3.5) circle (2pt) node[above,black,opacity=1] {\tiny{$q_1$}};
\fill [ blue] (2.5,2) circle (2pt) node[below,black,opacity=1] {\tiny{$s_1'$}};
\fill [ blue] (2.5,2.5) circle (2pt) node[above,black,opacity=1] {\tiny{$t_1'$}};
\fill [ black] (3.25,3.5) circle (2pt) node[below,black,opacity=1] {\tiny{$q$}};
\fill [ black] (3.5,2) circle (2pt) node[below,black,opacity=1] {\tiny{$\OP(s_1')$}};

\draw [round,fill=gray,opacity=0.2] (0.5,0) rectangle (1.5,1.5);

\end{tikzpicture}
\caption{Lemma \ref{lem:t2} shows that there exists a point $t_1'$ such that (1) it is the first point above $s_1'$ and (2) $^{q_1}\nwsearrow_{t_1'}$}
\label{fig:t_2-inside-mfc}
\end{figure}
	\begin{proof}
	Consider the rectangle $^{q_1}\Box_{s_1'}$. 
	By Lemma \ref{lem:greedy-property-left}, there exists 
	a point $t_1'$ to the left or above $s_1'$ in $^{q_1}\Box_{s_1'}$.
	However, by 
	Observation \ref{obs:s2obs2}, there is no point in (1) $(^{q_1}\Box_{s_1'})^{\circ}$ and 
	(2) below $q_1$ or to the left of $s_1'$ in $^{q_1}\Box_{s_1'}$. This implies that 
	$\abarrow{t_1'}{s_1'}$ and $t_1' \in {^{q_1}\Box_{s_1'}}$.

	We now need to show that $^{q_1}\nwsearrow_{t_1'}$.
	For contradiction, assume that $\lrarrow{q_1}{t_1'}$.
	This implies that
	$\OP(t_1') = q$ and $\lrarrow{t_1'}{q}$ 
	(since $\Key{s'_1} < \Key{r} < \Key{q}$ and $\abarrow{t_1'}{s_1'}$). 
	But we know by Observation \ref{obs:t2obs1} that $\OP(t_1')$ cannot lie to the right of $t_1'$. 
	Thus we have arrived at 
	a contradiction. So our assumption that $\lrarrow{q_1}{t_1'}$ must be false, 
	and hence $^{q_1}\nwsearrow_{t_1'}$.




	\end{proof}

	Our next claim is regarding the position of $\OP(t_1')$.

	\begin{lemma}
	\label{lem:t2location}
	Let $t_1'$ be the first point above $s_1'$. Then,
	$\OP(t_1') \ne t_1'$, and $^{\OP(t_1')}\nwsearrow_p$, $_{\OP(t_1')}\swnearrow^{q_1}$ and $^{\OP(t_1')}\nwsearrow_{s_1'}$.
	\end{lemma}
	\begin{proof}
	By Observation \ref{obs:t2obs1}, $\lrarrow{\OP(t_1')}{t'_1}$ or $\OP(t_1')=t_1'$. 
	However, $\OP(t_1') \ne t_1'$ because $r$ is the the original point in $_p\Box^q$
	with the smallest $x$-distance from $p$.
	This implies that $\lrarrow{\OP(t_1')}{t'_1}$.

	We have to show that $^{\OP(t_1')}\nwsearrow_p$. If $_p\neswarrow^{\OP(t_1')}$ (and $\lrarrow{\OP(t_1')}{t_1'}$), 
	it contradicts our assumption that $r$ is the the original point in $_p\Box^q$
	with the smallest $x$-distance from $p$ (recall that $\Key{r}>\Key{p'_1}=\Key{t'_1}$ by 
	Lemma~\ref{lem:p_2exists}). By Lemma \ref{lem:upperbox}, if $\RT(p)=B$, then
	there are no points 
	in $\UPB(B)$. Hence $^{\OP(t_1')}\nwsearrow_p$.


	This together with $\abarrow{t_1'}{s_1'}$ 
	implies $^{\OP(t_1')}\nwsearrow_{s_1'}$.
	Also since  $^{\OP(t_1')}\nwsearrow_p$, $_p\swnearrow^{q_1}$,
	 $^{q_1}\nwsearrow_{t_1'}$ (Lemma~\ref{lem:t2}) and $\lrarrow{\OP(t_1')}{t'}$, 
	hence $_{\OP(t_1')}\swnearrow^{q_1}$.
	\end{proof}
\fi
We say that a point $p_1$ is an {\em observable} point if $\AMmc(p_1)$ is not defined. 

\iflong
	\begin{lemma}
	\label{lem:consecutive-observable}
	Let $p_1 \leftrightarrow p_2$ be 
	observable points such that  $\CP(p_1),\CP(p_2) \in \RMFC_p$ for $p \in \XX$. 
	Assume that $\GRE$ put these points to arborally satisfy $_p\Box^{q_1}, {_p\Box^{q_2}}$, respectively,
	while processing point $p$.
	Then there is a point $t$ such that $_{t}\swnearrow^{q_1}$ 
	and $^{t}\nwsearrow_{q_{2}}$.
	\end{lemma}
	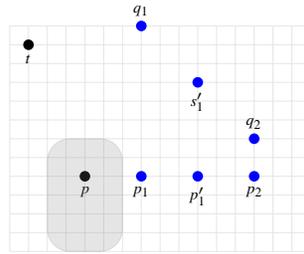
\begin{figure}[H]
\centering
\begin{tikzpicture}
[round/.style={rectangle, rounded corners=3mm, minimum size=20mm, draw=black!50}]

\draw[step=0.25,black,opacity=0.1] (0,0) grid (4,3);

\fill [black] (1,1) circle (2pt) node[below,black,opacity=1] {\tiny{$p$}};

\fill [ blue] (1.75,1) circle (2pt) node[below,black,opacity=1] {\tiny{$p_1$}};
\fill [ blue] (3.25,1) circle (2pt) node[below,black,opacity=1] {\tiny{$p_{2}$}};
\fill [ blue] (2.5,1) circle (2pt) node[below,black,opacity=1] {\tiny{$p_{1}'$}};
\fill [ blue] (2.5,2.25) circle (2pt) node[below,black,opacity=1] {\tiny{$s_{1}'$}};

\fill [ blue] (1.75,3) circle (2pt) node[above,black,opacity=1] {\tiny{$q_1$}};
\fill [ blue] (3.25,1.5) circle (2pt) node[above,black,opacity=1] {\tiny{$q_{2}$}};

\fill [black] (0.25,2.75) circle (2pt) node[below,black,opacity=1] {\tiny{$t$}};

\draw [round,fill=gray,opacity=0.2] (0.5,0) rectangle (1.5,1.5);

\end{tikzpicture}
\caption{The setting of Lemma \ref{lem:consecutive-observable}: $p_1$ and $p_{2}$ are two {\em observable} points. The above figure depicts the case when the next point to the right of $p_1$, $p_1' \ne  p_{2}$.}
\label{fig:consecutive-observable}
\end{figure}
	\begin{proof}
	One can check that $\Time{q_1} <\Time{q_2}$.
	Let $p_1'$ be the first point to the right of $p_1$. And let $s_1'$ be the first point 
	above $p_1'$ (see Fig.~\ref{fig:consecutive-observable} for an illustration). 
	Note that if $p_1' = p_{2}$, then $s_1'= q_{2}$.
	Else if  $\lrarrow{\lrarrow{p_1}{p_1'}}{p_2}$, then one can check that 
	$^{s_1'}\senwarrow_{q_{2}}$.
	$\GRE$ put $p_1'$ and $p_{2}$ due to the unsatisfied rectangles $_p\Box^{s_1'}$
	and $_p\Box^{q_{2}}$, respectively. By Lemma~\ref{lem:t2location}, there
	exists a $t$ such that $_{t}\swnearrow^{q_1}$ 
	and $^{t}\nwsearrow_{s_1'}$. We can replace $s_1'$ with $q_{2}$
	in the previous statement as either $s_1'=q_2$ or $^{s_1'}\senwarrow_{q_{2}}$.


	\end{proof}

	Continuing the setting that was used in Lemma \ref{lem:consecutive-observable},
	let $\CP(p_i), \CP(p_{i+1}) \in \RMFC_p$.
	If $\RT(p)=B$, then
	by Lemma~\ref{lem:reversetop}, $\OP(q_i)$ is in block $B_i$, where 
	$B_i \in \SB(B)$. Similarly, $\OP(q_{i+1}) \in B_{i+1}$
	such that $B_{i+1} \in \SB(B)$.
	We now prove the following lemma that is similar to Lemma \ref{lem:finalmmc}.

\fi

\iflong
	\begin{lemma}
	\label{lem:finalmfc}
\else
	\newtheorem*{lem:finalmfc}{Lemma \ref{lem:finalmfc}}
	\begin{lem:finalmfc}
\fi
Let $p \in \XX$ and $\RT(p) = B$.
Let $p_1 \leftrightarrow p_2 \leftrightarrow \dots \leftrightarrow p_\ell$ be 
observable points such that $\CP(p_i) \in \RMFC_p$ for $i \in [\ell]$.
Assume that $\GRE$ puts these points due to arborally unsatisfied
rectangles $_p\Box^{q_1}, {_p\Box^{q_2}}, \dots, {_p\Box^{q_\ell}}$, respectively, while processing point $p$.
By Lemma \ref{lem:reversetop}, $\OP(q_i) \in B_i$ such that $B_i \in \SB(B)$.
 Then $B_i \ne B_j$ for $i \neq j$. 
\iflong
	\end{lemma}
\else	
	\end{lem:finalmfc}
\fi
\iflong
	\begin{proof}
	Consider points $p_i \leftrightarrow p_j$. 
	Since $\CP(p_i), \CP(p_{j}) \in \RMFC_p$, we have $\lrarrow{q_i}{\OP(q_i)}$ and $\lrarrow{q_j}{\OP(q_{j})}$ 
	($\OP(q_i) \ne q_i$ as otherwise $\RR(p_j) \in \GR$; similarly for $\OP(q_j)$).
	By Lemma~\ref{lem:consecutive-observable}, there is a point $t$ such that
	$_{t}\swnearrow^{q_i}$ 
	and $^{t}\nwsearrow_{q_{j}}$. This implies that $_{t}\swnearrow^{\OP(q_i)}$ 
	and $^{t}\nwsearrow_{\OP(q_{j})}$. 
	Moreover, by Lemma~\ref{lem:t2location} we have $^{t}\nwsearrow_p$. 
	Hence $\Key{t} <  \Key{p} <\Key{\OP(q_i)}$, and
	$p \notin B_i$ (since $p \in B$ and $B_i$ is the sibling of $B$). Hence, $t \notin B_i$ as all
	the points in block $B_i$ should be contiguous in their keys. 

	Since $_{t}\neswarrow^{\OP(q_i)}$ and $^{t}\nwsearrow_{\OP(q_{j})}$, we also have 
	$\Time{\OP(q_i)} <  \Time{t} < \Time{\OP(q_{j})}$. 
	Since $t \notin B_i$,  $\OP(q_i)$ and $\OP(q_{j})$ cannot be in the same
	block, as all the points in a block should be contiguous in time. Hence $B_i \neq B_j$. 
	\end{proof}
\fi

\subsection{Proof of Lemma~\ref{lem:bad-mfc-bound}}
Since there are at most $k-1$ siblings of $B$, by Lemma~\ref{lem:finalmfc} for each $p \in \XX$ we have
$|\{ \CP(p_i) \in \RMFC_p \mid p_i \text{ is an observable point}\}| \le k-1$.
By symmetry, $|\{ \CP(p_i) \in \LMFC_p \mid p_i \text{ is an observable point}\}| \le k-1$ for 
$p \in \XX$. 
So, $|\{ \CP(p_i) \in \BR_p \mid p_i$ is an observable point$\}| \le 2(k-1)$ for $p \in \XX$.
We are now ready to upper bound the number of points in $\BR$.
\begin{align*}
|\BR| &= \sum_{p \in \XX} \Big(|\{ \CP(p_i) \in \BR_p\mid p_i \in \mathrm{domain}(\AMmc) \}| 
+ |\{ \CP(p_i) \in \BR_p \mid p_i\text{ is an observable point}\}|\Big) \\
& \le 4|\MMC| + 2n(k-1).
\end{align*}

Theorem \ref{thm:mmcfinalbound} gives 
$|\MMC| \leq 2n(k-1)$. Hence $|\BR| \le  10 n(k-1)$, proving Lemma~\ref{lem:bad-mfc-bound}.

\section{Improving the bound on the number of observable points in $\MFC$}
\label{sec:mfc-better}
Improving the bound in Lemma~\ref{lem:bad-mfc-bound}, we prove 
\begin{lemma}\label{lem:bad-mfc-bound-improved}
$|\BR| \le  O(n\log{k})$.
\end{lemma}

We again define the notations used in Section \ref{sec:mmc-betterbound}.
Let $B_1, B_2, \dots, B_l$ be the children of $B$
in $\TD(\XX)$. Let $B_i, B_{i+1},\dots,B_{i+2m-1}$ be the consecutive 
$2m$ children of $B$. 
Let $\BB_{\ell} = \{B_i,B_{i+1},\dots,B_{i+m-1}\}$ and $\BB_r = \{B_{i+m}, B_{i+m+1}, \dots, B_{i+2m-1}\}$.
Let  $N_{\Top(B_j)}$ and  $\R_{\Top(B_j)}$ be the set of key-new and key-old elements 
added by $\GRE$ in $\RG(\BB_r)$ while processing $\Top(B_j)$.
In Section \ref{sec:proofpartitionlemma}, we proved that $\sum_{B_j \in \BB_{\ell}} |R_{\Top(B_j)}| \le 12m$.
Note that these points can be in $\MMC$ or $\MFC$. In this section, we need to bound the number of
observable points in $N_{\Top(B_j)}$. To this end, we will prove the following lemma which is similar to 
lemma \ref{lem:partitionlemma}.

\begin{lemma}
\label{lem:mfcpartitionlemma}
Let $B_i,B_{i+1},\dots, B_{i+2m-1}$ be the set of consecutive children blocks of $B$ in $\TD(\XX)$.
The number of observable points added by $\{ \Top(B_{i}), \Top(B_{i+1}), \dots, \Top(B_{i+m-1}) \} \setminus \Top(B)$ 
in \\$\RG(B_{i+m}, B_{i+m+1}, \dots, B_{i+2m-1})$  is $O(m)$.
\end{lemma}
\begin{proof}

Consider the points in $N_{\Top(B_j)}$ where $B_j \in \BB_{\ell}$. 
By Lemma \ref{lem:key-newinmfcsymmetric},
these points have $\RR(\cdot) \in \MFC$. 
Let $O_{\Top(B)}$
be the set observable points in $N_{\Top(B_j)}$. We will show that 
$\sum_{B_j \in \BB_{\ell}} |O_{\Top(B)}| \le |\{B_k \in \BB_r | \Top(B_j) \in \LeftRel(B_k) \}| $.

Consider any block $B_k \in \BB_r$. By Corollary \ref{cor:pointotherthanrelimplication}, 
only points in $\LeftRel(B_k)$ and $\RightRel(B_k)$ can
add key-new points in $\RG(B_k)$. Let us assume that $\Top(B_j) \in \LeftRel(B_k)$ and assume that 
$\GRE$ adds two observable point $\lrarrow{p_1}{p_2}$ in $\RG(B_k)$ while processing $\Top(B_j)$
due to unsatisfied rectangle $_{\Top(B_j)}\Box^{q_1}, _{\Top(B_j)}\Box^{q_2}$
respectively. Since $p_1$ and $p_2$ are key-new points, $q_1, q_2$ lie in $\BOX(B_k)$.
So $\OP(q_1), \OP(q_2) \in B_k$. However, this contradicts Lemma \ref{lem:finalmfc} which states 
that $\OP(q_1)$ and $\OP(q_2)$ are the elements of different siblings of $B_k$.

So, the number or observable points $|O_{\Top(B_j)}|$ can be bounded by: $|\{B_k \in \BB_r | \Top(B_j) \in \LeftRel(B_k) \}|$.
And, $\sum_{B_j \in \BB_{\ell}} |O_{\Top(B)}| \le \sum_{B_j \in \BB_{\ell}}|\{B_k \in \BB_r | \Top(B_j) \in \LeftRel(B_k) \}|$.
Since there are at most 2 left-relative of any block, the above inequality can be written as:
$\sum_{B_j \in \BB_{\ell}} |O_{\Top(B)}| \le \sum_{B_k \in \BB_r} 2 = 2m$. 

Since, points in $R_{\Top(\cdot)}$ may also be observable, the total number of observable points added in $\RG(\BB_r)$ while
processing points in $\{ \Top(B_i), \Top(B_{i+1}), \dots, \Top(B_{i+m-1}) \} \setminus \Top(B) = \sum_{B_j \in \BB_{\ell}} (R_{\Top(B_j)} + O_{\Top(B_j)}) \le 12m +2m = 14m$.
\end{proof}

Consider a block $B$ in $\TD(\XX)$ having children $B_1, B_2, \dots, B_{\ell} (\ell \le k)$.
Let $Y(B) := T(B_1,B_2,\dots,B_{\ell})$ be the total number of observable points added by 
$\GRE$ while processing points in $\{ \Top(B_1), \Top(B_2),\dots, \Top(B_{\ell}) \} \setminus \Top(B)$ in $\RG(B_1,B_2,\dots, B_{\ell})$. 
Then using lemma \ref{lem:mfcpartitionlemma} and its symmetric version, we can calculate $Y(B)$ as follows:
$Y(B) = T(B_1,B_2,\dots,B_{\ell})$ = $T(B_1,B_2,\dots,B_{\ell/2})$ + $T(B_{\ell/2+1},B_{\ell/2+2},\dots,B_{\ell})$ + $14 \ell$.
This would imply that $Y(B) = T(B_1,B_2,\dots,B_{\ell}) \le 14 \ell \log \ell$. 

We would charge these $O(\ell \log \ell)$ points 
to the following $\ell-1$ original points: $\{ \Top(B_1), \Top(B_2), \dots, \Top(B_{\ell}) \}\setminus {\Top(B)}$.
That is, each top point of children of block $B$ except one gets $15 \log \ell$ charge.
Similar to Theorem \ref{thm:boundmmc}, we can bound the number of observable points.

\begin{theorem}
\label{thm:boundmfc}
$\sum_{p \in \XX} |\{ \CP(p_i) \in \BR_p \mid p_i\text{ is an observable point}\}| = 16 n \log k$
\end{theorem}
\begin{proof}
Let $\RT(p) = B'$ and $\PP(B') = B$. By Lemma \ref{lem:blockleftright}, $\GRE$ can put at most two point below $\Left(B)$
and $\Right(B)$. And by the analysis above, the amortized number of observable points added by $\GRE$ in 
$\BOX(B)$ while processing $p$ is $15 \log \ell = 15 \log k$ where $\ell$ is the number of children of $B$. 
So, 
amortized number of points added by $p = 2 + 15 \log k$. So $\sum_{p \in \XX} |\{ \CP(p_i) \in \BR_p \mid p_i\text{ is an observable point}\}| \le 16 n \log k$.
\end{proof}

\subsection{Proof of Lemma~\ref{lem:bad-mfc-bound-improved}}
We are now ready to upper bound the number of points in $\BR$.
\begin{align*}
|\BR| &= \sum_{p \in \XX} \Big(|\{ \CP(p_i) \in \BR_p\mid p_i \in \mathrm{domain}(\AMmc) \}| 
+ |\{ \CP(p_i) \in \BR_p \mid p_i\text{ is an observable point}\}|\Big) \\
& \le 4|\MMC| + 16 n \log k.
\end{align*}

Theorem \ref{thm:mmcfinalbound} gives 
$|\MMC| \leq 14 n \log k$. Hence $|\BR| \le  80 n \log k$, proving Lemma~\ref{lem:bad-mfc-bound-improved}.

\section{Proof of the main result}
\label{sec:main}
With all the ingredients at hand, the proof of the main result Theorem~\ref{thm:main} is now short:
\iflong
\begin{align*}
	|\GG| &\leq  2 |\CP(\GG)|  & \text{(Corollary~\ref{cor:couplingY})}\\
		  &=  2(|\MMC| + |\MFC|) & \text{(Lemma~\ref{lem:copulingsamecardinality})}\\
		  &=  2(|\MMC| + |\BR| + |\GR|) \\
	&\leq 188 n \log k + 2|\GR|  & \text{(Theorem~\ref{thm:mmcfinalbound} and Lemma~\ref{lem:bad-mfc-bound})}\\
	&\leq 188 n \log k + 4|\XX \cup \OPT(\XX)|. & \text{(Theorem~\ref{lem:goodbound})}
	\end{align*}
\fi

Using Corollary 1.10 in \cite{ChalermsookG0MS15b}, namely
$|\XX \cup \OPT(\XX)| =  O(n \log k)$ for $k$-decomposable sequences, we get $|\GG| = O(n \log{k})$, which immediately gives
Theorem~\ref{thm:main}.

\bibliographystyle{plain}
\nocite{*}
\bibliography{arboral}
\end{document}